\documentclass[a4paper,11pt,oneside]{article}

\usepackage[left=1.0in, right=1.0in, top=1.0in, bottom=1.0in]{geometry}

\usepackage[toc, page]{appendix}
\usepackage{amsmath}
\usepackage{amsfonts}
\usepackage{amssymb}
\usepackage{amsthm}
\usepackage{mathtools}
\usepackage{color}
\usepackage{graphics}
\usepackage{graphicx}
\usepackage{epsfig}
\usepackage{epstopdf}
\usepackage{float}
\usepackage{rotating}
\usepackage{array}
\usepackage{textcomp}
\usepackage{bbm}
\usepackage{fancyref}
\usepackage{hyperref}
\usepackage{xcolor}
\usepackage{rotating}
\usepackage[noadjust]{cite} % basic fos bibliography [1]-[5]
\usepackage{tikz}
\usetikzlibrary{calc,positioning,shapes,shadows,arrows,fit}
\usepackage{algorithm}
\usepackage{algorithmic}
\usepackage[font = footnotesize]{caption}
\usepackage{wrapfig}

\newtheorem{theorem}{Theorem}
\newtheorem{lemma}{Lemma}

\theoremstyle{definition}
\newtheorem{proposition}{Proposition}
\newtheorem{definition}{Definition}
\newtheorem{remark}{Remark}
\newtheorem{assumption}{Assumption}

\newtheorem{problem}{Problem}

\title{Robust Tube-based Model Predictive Control for \\ Time-constrained Robot Navigation}
		
\author{Alexandros Nikou and Dimos V. Dimarogonas%
\thanks{The authors are with the School of Electrical Engineering and Computer Science, KTH Royal Institute of Technology, Stockholm, Sweden. E-mail: {\tt\small \{anikou,dimos\}@kth.se}. This work was supported by the H2020 ERC Grant BUCOPHSYS, the EU H2020 Co4Robots project, the Swedish Foundation for Strategic Research (SSF), the Swedish Research Council (VR) and the Knut och Alice Wallenberg Foundation (KAW).}}

\date{}

\begin{document}
\maketitle

\begin{abstract}
This paper deals with the problem of time-constrained navigation of a robot modeled by uncertain nonlinear non-affine dynamics in a bounded workspace of $\mathbb{R}^n$. Initially, we provide a novel class of robust feedback controllers that drive the robot between Regions of Interest (RoI) of the workspace. The control laws consists of two parts: an on-line controller which is the outcome of a Finite Horizon Optimal Control Problem (FHOCP); and a backstepping feedback law which is tuned off-line and guarantees that the real trajectory always remains in a bounded hyper-tube centered along the nominal trajectory of the robot. The proposed controller falls within the so-called tube-based Nonlinear Model Predictive control (NMPC) methodology. Then, given a desired high-level specification for the robot in Metric Interval Temporal Logic (MITL), by utilizing the aforementioned controllers, a framework that provably guarantees the satisfaction of the formula is provided. The proposed framework can handle the rich expressiveness of MITL in both safety and reachability specifications. Finally, the proposed framework is validated by numerical simulations. \\

\end{abstract}

\section{Introduction}
%\vspace{-1mm}
\emph{Navigation} is an important field in both the robotics and the control communities, due to the need for autonomous control \cite{rimon1992exact}. Applications arise in the fields of autonomous driving and air-traffic management. In another line of research, the problem of controlling systems under \emph{high-level specifications} has been gaining significant research attention \cite{belta2008fully, fainekos_girard_2009_temporal, murray_2012_ltl}. The aim of this work is to address a robot navigation problem under high-level tasks which include both time constrained reachability and safety. 

The qualitative specification language that has primarily been used to express the high-level tasks is Linear Temporal Logic (LTL) \cite {baier2008principles}. A suitable temporal logic for dealing with tasks that are required to be completed within \emph{certain time bounds} is \emph{Metric Interval Temporal Logic (MITL)}. MITL has been originally proposed in \cite{alur_1994} and has been used for controller synthesis in \cite{alex_2016_acc, alex_automatica_2018}. Given robot dynamics and an MITL formula, the controller synthesis procedure is as follows: first, the robot dynamics are abstracted into a discrete representation, the so-called Weighted Transition System (WTS), in which the time duration for navigating between states is modeled by a weight in WTS (abstraction); second, a product between the WTS and an automaton that accepts the runs that satisfy the given formula is computed; and third, once an accepting run in the product found, it maps into a sequence of feedback controllers of the robot dynamics.

The main focus of this paper is the first part of the aforementioned procedure. In particular, we aim to provide an \emph{abstraction of robot dynamics} into WTS in such a way that the rich MITL expressiveness in both reachability and safety specifications is exploited. In order address this problem, and due to the fact that the robot dynamics are highly nonlinear, under state/input constraints as well as under the presence of external disturbances/unmodeled dynamics, a Nonlinear Model Predictive Control (NMPC)~framework is~used~\cite{michalska_1993, frank_1998_quasi_infinite, mayne_2000_nmpc, grune2016nonlinear}. 

One of the main challenges in NMPC is the efficient handling of the external disturbances/uncertainties. A promising robust strategy, originally proposed for discrete-time linear systems in \cite{rakovic_2004_tubes_1}, is the so called tube-based approach. Tube-based approaches for  affine in the control continuous-time nonlinear systems with constant matrices multiplying the control input vectors have been proposed in \cite{yu_2013_tube}, which we aim to extend here in order to cover a larger class of nonlinear systems. Preliminary results in tube-based NMPC for nonlinear non-affine systems that state and input space have the same dimension can be found in our earlier work \cite{alex_IJRNC_2018}. In the current paper, these results are extended to underactuated systems, which arise in robotic applications with Lagrangian kinematics/dynamics models and their controller design constitutes a challenging task.

\noindent By taking into consideration the aforementioned, the contribution of this paper is twofold:
\begin{itemize}
	\item given a robot modeled by uncertain nonlinear non-affine dynamics and a workspace with RoI, under standard NMPC and controllability assumptions, a systematic control design methodology for tube-based NMPC which guarantees robust navigation between RoI under safety constraints is developed;
	%\item the computation of the navigation time that the robot requires to be navigated between the RoI is explicitly provided;
	\item we exploit the aforementioned control design in order to abstract the dynamics of the robot into a WTS. Then, given an MITL formula that the robot needs to satisfy for all times, by performing an MITL control synthesis procedure, a sequence of control laws that guarantees the satisfaction of the formula is provided. 
\end{itemize}

This constitutes a novel solution to a time-constrained navigation problem for systems with general uncertain dynamics and input/state constraints, in which complex tasks that include both reachability and safety are imposed.

The rest of this manuscript is structured as follows: Section \ref{sec:notation_preliminaries} provides the notation that will be used as well as necessary background knowledge; In Section \ref{sec:problem_formulation}, the problem treated in this paper is formally defined; Section \ref{sec:main_results} contains the main results of the paper; Section \ref{sec:simulation_results} is devoted in numerical simulations; and in Section \ref{sec:conclusions}, conclusions and future research directions are discussed.

\section{Notation and Preliminaries} \label{sec:notation_preliminaries}

The sets of positive integers, positive rational numbers and real numbers are denoted by $\mathbb{N}$, $\mathbb{Q}_{+}$ and $\mathbb{R}$, respectively. Given a set $\mathcal{S}$, denote by $|\mathcal{S}|$ its cardinality, by $\mathcal{S}^N = \mathcal{S} \times \dots \times \mathcal{S}$ its $N$-fold Cartesian product and by $2^\mathcal{S}$ the set of all its subsets. Given a vector $y \in \mathbb{R}^n$ denote by $\|y\|_{\scriptscriptstyle 2} \coloneqq \sqrt{y^\top y}$ and $\|y\|_{\scriptscriptstyle B}$ $\coloneqq \sqrt{y^\top B y}$, $B $ $\ge 0$ its Euclidean and weighted norm, respectively; $\lambda_{\scriptscriptstyle \min}(B)$ stands for the minimum absolute value of the real part of the eigenvalues of $B \in \mathbb{R}^{n \times n}$; $0_{m \times n} \in \mathbb{R}^{m \times n}$ and $I_n \in \mathbb{R}^{n \times n}$ stand for the $m \times n$ matrix with all entries zeros and the identity matrix, respectively. The notation ${\rm diag}\{B_1$, $\dots$, $B_n\}$ stands for the block diagonal matrix with the matrices $B_1$, $\dots$, $B_n$ in the main diagonal; $\mathcal{B}(\chi, r) \coloneqq \{y \in \mathbb{R}^n: \|y-\chi\|_{2} \leq r\}$ stands for the $n$-th dimensional ball with center and radius $\chi \in \mathbb{R}^{n}$, $r >0$, respectively. Given a function $f: \mathbb{R}^n \to \mathbb{R}^m$, $\frac{\partial f_i}{\partial x_j}$ denotes the element of row $i$ and column $j$ of the Jacobian matrix of $f$, with $i$ $\in \{1,\dots,n\}$ and $j \in$ $\{1,\dots,m\}$. Given two vectors $x$, $y \in \mathbb{R}^n$ their \emph{convex hull} is defined by: $$\mathcal{C}(x,y) \coloneqq \{\theta x + (1-\theta) y : \theta \in (0,1)\}.$$ A vector of a canonical basis of $\mathbb{R}^n$ is defined by:
\begin{align} \label{eq:basis_vector}
\mathfrak{e}_n^{i} \coloneqq \Big[0, \dots, 0, \underbrace{1}_{\scriptscriptstyle i-\rm{th \ element}}, 0, \dots, 0 \Big]^\top.
\end{align}
Given~sets~$\mathcal{S}_1$, $\mathcal{S}_2$~$\subseteq \mathbb{R}^n$~and~matrix $B \in \mathbb{R}^{n \times m}$,~the \emph{Minkowski addition}, the~\emph{Pontryagin~difference} and the \emph{matrix-set multiplication} are respectively defined by: 
\begin{align*}
\mathcal{S}_1 & \oplus \mathcal{S}_2 \coloneqq \{s_1 + s_2 : s_1 \in \mathcal{S}_1, s_2 \in \mathcal{S}_2\}, \\
\mathcal{S}_1 & \ominus \mathcal{S}_2 \coloneqq \{s_1 : s_1+s_2 \in \mathcal{S}_1, \forall s_2 \in \mathcal{S}_2\}, \\
B & \circ \mathcal{S} \coloneqq \{b: \exists s \in \mathcal{S}, b = Bs\}.
\end{align*}

\begin{lemma} \cite{alex_IJRNC_2018} \label{lemma:basic_ineq}
For any constant $\rho > 0$, vectors $y_1$, $y_2 \in \mathbb{R}^n$ and matrix $B \in \mathbb{R}^{n \times n}$, $B > 0$ it holds that: $$y_1 B y_2 \le \frac{1}{4 \rho} y_1^\top B y_1 + \rho y_2^\top B y_2.$$
\end{lemma}
\begin{proposition} (Mean value theorem for vector valued functions \cite{zemouche2008observers}) \label{pro:mvt}
	Consider a function $f: \mathbb{R}^n \to \mathbb{R}^m$ which is differentiable on an open set $\mathcal{S} \subseteq \mathbb{R}^n$. Let $x$, $y$ two points of $\mathcal{S}$ such that $\mathcal{C}(x,y) \subseteq \mathcal{S}$. Then, there exist constant vectors $\varpi_1$, $\dots,$ $\varpi_m \in \mathcal{C}(x,y)$ such that:
	\begin{align}
	f(x)-f(y) = \left[ \sum_{k = 1}^{m} \sum_{j=1}^{n} \mathfrak{e}_m^k (\mathfrak{e}_n^j)^\top \frac{\partial f_k(\varpi_k)}{\partial x_j} \right] (x-y).
	\end{align}
\end{proposition}
\begin{definition} \label{def:RPI_set} \cite{alex_IJRNC_2018}
Consider a dynamical system $\dot{\chi} = f(\chi,u,d)$ where: $\chi \in \mathcal{X}$, $u \in \mathcal{U}$, $d \in \mathcal{D}$ with initial condition $\chi(0) \in \mathcal{X}$. A set $\mathcal{X}' \subseteq \mathcal{X}$ is a \emph{Robust Control Invariant (RCI) set} for the system, if there exists a feedback control law $u \coloneqq \kappa(\chi) \in \mathcal{U}$, such that for all $\chi(0) \in \mathcal{X}'$ and for all $d(t) \in \mathcal{D}$ it holds that $\chi(t) \in \mathcal{X}'$ for all $t \ge 0$, along every solution $\chi(t)$.
\end{definition}

\begin{definition} \label{def: WTS} \cite{alex_2016_acc}
A \emph{Weighted Transition System} (WTS) is a tuple $(S, S_0, {\rm Act}, \longrightarrow, \mathfrak{t}, \Sigma, L)$ where $S$ is a finite set of states; $S_0 \subseteq S$ is a set of initial states; ${\rm Act}$ is a set of actions; $\longrightarrow \subseteq S \times {\rm Act} \times S$ is a transition relation; $\mathfrak{t}: \longrightarrow \rightarrow \mathbb{Q}_{+}$ is a map that assigns a positive weight to each transition; $\Sigma$ is a finite set of atomic propositions (an \textit{atomic proposition} $\sigma \in \Sigma$ is a statement that is either true or false); and $L: S \rightarrow 2^{\Sigma}$ is a labeling function.
\end{definition}

\begin{definition}\label{run_of_WTS} \cite{alex_2016_acc}
A \textit{timed run} of a WTS is an infinite sequence $r^t = (r(0), \tau(0))(r(1), \tau(1)) \ldots$, such that $r(0) \in S_0$, and for all $l \geq 0$, it holds that $r(l) \in S$ and $(r(l), u(l), r(l+1)) \in \longrightarrow$ for a sequence of actions $u(0) u(1) u(2) \ldots$ with $u(l) \in \rm Act$, $\forall l \geq 0$. The \textit{time stamps} $\tau(l)$, $l \geq 0$ are inductively defined as: $1)$ $\tau(0) = 0$; $2)$ $\displaystyle \tau(l+1) \coloneqq  \tau(l) + \mathfrak{t}(r(l), r(l+1))$, $\forall l \geq 0$.
\end{definition}

\noindent The syntax of MITL (see \cite{alur_1994}) over a set of atomic propositions $\Sigma$ is defined by the grammar:
\begin{equation*}
\varphi := \sigma \ | \ \neg \varphi \ | \ \varphi_1 \wedge \varphi_2 \ | \ \bigcirc_I \varphi  \ | \ \Diamond_I \varphi \mid \square_I \varphi \mid  \varphi_1 \ \mathcal{U}_I \ \varphi_2,
\end{equation*}
where $\sigma \in \Sigma$, and $\bigcirc$, $\Diamond$, $\square$ and $\mathcal{U}$ are the next, eventually, always and until temporal operator, respectively; $\neg$, $\wedge$ are the negation and conjunction operators, respectively; $I = [a, b] \subseteq \mathbb{Q}_{+}$ where $a, b \in [0, \infty]$ with $a < b$ is a non-empty timed interval. MITL formulas are interpreted over timed runs like the ones produced by a WTS which is given in Definition \ref{run_of_WTS}. For the semantics of MITL see \cite[Sec. II, p. 2]{alex_2016_acc}. Any MITL formula $\varphi$ over $\Sigma$ can be algorithmically translated into a Timed B\"uchi Automaton (TBA) with the alphabet $2^{\Sigma}$, such that the language of timed words that satisfy $\varphi$ is the language of timed words produced by the TBA \cite{maler_MITL_TA}. Due to space constraints, for a detailed preliminary background regarding timed verification we refer the reader to \cite{alex_2016_acc}.
\vspace{-1mm}
\section{Problem Formulation} \label{sec:problem_formulation}
\subsection{System Model}
Consider a robot operating in a workspace $\mathcal{W} \subseteq \mathbb{R}^n$ governed by the following \emph{uncertain kinematics-dynamics model}:
\vspace{-3mm}
\begin{subequations}
	\begin{align} 
	\dot{\chi} & = v, \label{eq:kinematics} \\
	\dot{v} & = f(\chi,v,u) + d, \label{eq:unc_dynamics}
	\end{align}
\end{subequations}
where $\chi \in \mathcal{W}$ denotes the state of the robot in the workspace (position, orientation); $v \in \mathbb{R}^{n}$ stands for the velocity; $u \in \mathbb{R}^n$ is the control input; $f: \mathbb{R}^n \times \mathbb{R}^n \times \mathbb{R}^n \to \mathbb{R}^n$ is a continuous nonlinear function; and $d \in \mathbb{R}^n$ stands for external disturbances, uncertainties and unmodeled dynamics. The velocity is constrained in a connected set $\mathcal{V} \subseteq \mathbb{R}^{n}$ which contains the origin. The control inputs need to satisfy $u \in \mathcal{U}$, where $\mathcal{U}$ is a convex set containing the origin. Consider also bounded disturbances $d \in \mathcal{D} \coloneqq \{d \in \mathbb{R}^n: \|d\|_{\scriptscriptstyle 2} \le \widetilde{d}\}$, where $\widetilde{d} > 0$. Define the corresponding \emph{nominal dynamics} for \eqref{eq:kinematics}-\eqref{eq:unc_dynamics} by:
\begin{subequations}
	\begin{align} 
	\dot{\overline{\chi}} & = \overline{v}, \label{eq:nom_kinematics} \\
	\dot{\overline{v}} & = f(\overline{\chi}, \overline{v}, \overline{u}), \label{eq:nom_dynamics}
	\end{align}
\end{subequations}
where $w \equiv 0$, $\overline{\chi} \in \mathcal{W}$, $\overline{v} \in \mathcal{V}$ and $\overline{u} \in \mathcal{U}$. Define $J: \mathcal{W} \times \mathcal{V} \times \mathcal{U} \to \mathbb{R}^n \times \mathbb{R}^{n}$ by:
\begin{align} \label{eq:funct_J}
J(\chi,v,u) = \sum_{i = 1}^{n} \sum_{j=1}^{n} \mathfrak{e}_n^i (\mathfrak{e}_n^j)^\top \frac{\partial f_{i}(\chi,v,u)}{\partial u_j},
\end{align}
where $f_i$ is the $i$-th component of the vector-valued function $f$, and $\mathfrak{e}_n^{i}$, $\mathfrak{e}_n^{j}$ as given in \eqref{eq:basis_vector}.

\begin{assumption} \label{ass:f_diff}
	$f$ is \emph{continuously differentiable} with respect to $x$, $v$ and $u$ in $\mathcal{W} \times \mathcal{V} \times \mathcal{U}$ with $f(0,0,0) = 0$.
\end{assumption}

\begin{assumption} \label{ass:stabiliz}
	The linear system $\dot{\overline{\eta}} = A \overline{\eta} + B \overline{u}$, where $\overline{\eta} \coloneqq [\overline{\chi}^\top, \overline{v}^\top]^\top \in \mathbb{R}^{2n}$, that is the outcome of the Jacobian linearization of the nominal dynamics \eqref{eq:nom_kinematics}-\eqref{eq:nom_dynamics} around the equilibrium state $(\chi, v) = (0,0)$ is stabilizable.
\end{assumption}

\begin{assumption} \label{ass:J_lower_bound}
There exists a constant $\underline{J}$ such that:
\begin{align} \label{eq:J_lower_bound}
\lambda_{\min}\left[\frac{J(\cdot) + J(\cdot)^\top}{2}\right] \ge \underline{J} > 0, \forall \chi \in \mathcal{W}, v \in \mathcal{V}, u \in \mathcal{U}. 
\end{align}
\end{assumption}

In the given workspace, there exist $\mathfrak{m} \in \mathbb{N}$ Regions of Interest (RoI) labeled by $\mathcal{M} \coloneqq \{1,\dots, \mathfrak{m}\}$. Without loss of generality, assume that the RoI are modeled by balls, i.e., $\mathcal{R}_m \coloneqq \mathcal{B}(y_{m}, p_m)$, $m \in \mathcal{M}$, where $y_{m}$ and $p_m > 0$ stands for the center and radius of RoI $\mathcal{R}_m$, respectively. Define also the the union of RoI by $$\mathcal{R} \coloneqq \bigcup_{m \in \mathcal{M}} \mathcal{R}_m.$$ 

Due to the fact that we are interested in imposing safety specifications to the robot, at each time $t \ge 0$, the robot is occupying a ball $\mathcal{B}(\chi(t), \mathfrak{r})$ that covers its volume, where $\chi(t)$ and $\mathfrak{r} > 0$ are its center and radius, respectively. Assume that $\displaystyle \min_{m \in \mathcal{M}} \{p_m\} > \mathfrak{r}$, which means that the RoI have sufficiently larger volume than the robot. 

\subsection{Objectives}

The control objective is the navigation of the robot with dynamics as in \eqref{eq:kinematics}-\eqref{eq:unc_dynamics} between RoI so that it obeys a given high-level timed specification over atomic tasks. Atomic tasks are captured through a given finite set of atomic propositions $\Sigma$. Each RoI is labeled with atomic propositions that hold true there. Define the labeling function:
\begin{align} \label{eq:label_function_L}
L: \mathcal{R} \to 2^{\Sigma},
\end{align}
which maps each RoI with a subset of atomic propositions that hold true there. Note that some of the RoI may be assigned with labels that indicate \emph{unsafe regions}, i.e., the robot is required to avoid visiting them (\emph{safety specifications}).

\begin{definition} \label{def:unique_timed_word}
	A trajectory $x(t)$ is \emph{uniquely associated with a timed run} $r^t = (r(0), \tau(0))$ $(r(1), \tau(1))(r(2), \tau(2))\ldots$, where $r(l) \in \mathcal{R}$, $\forall l \in \mathbb{N}$, is a sequence of RoI that the robot crosses, if the following hold:	
	\begin{enumerate}
		\item $\tau(0) = 0$, i.e., the robot starts the motion at time $t = 0$;
		\item $\mathcal{B}(x(\tau(0)), \mathfrak{r}) \subsetneq r(0)$, i.e., initially, the volume of the robot is entire within the RoI $r(0) \in \mathcal{R}$;
		\item $\mathcal{B}(x(\tau(l)), \mathfrak{r}) \subsetneq r(l)$, $\forall l \in \mathbb{N}$, i.e., the robot changes discrete state \emph{only when} its entire volume is contained in the corresponding RoI;
		\item $\tau(l+1) \coloneqq \tau(l) + \mathfrak{t}(r(l), r(l+1))$, $\forall l \in \mathbb{N}$, where:
		\begin{equation} \label{eq:desired_times_T}
		\mathfrak{t}: \mathcal{R} \times \mathcal{R} \to \mathbb{Q}_{+},
		\end{equation}
		is a function that models the duration that the robot needs to be driven between regions $r(l)$ and $r(l+1)$.
	\end{enumerate}
\end{definition}  

\begin{definition} \label{def:x_satisfaction}
	A trajectory $\chi(t)$ \emph{satisfies} an MITL formula $\varphi$ over the set of atomic propositions $\Sigma$, formally written as $\chi(t) \models \varphi$, $\forall t \ge 0$, if and only if there exists a timed run $r^t$ to which the trajectory $\chi(t)$ is uniquely associated, according to Definition \ref{def:unique_timed_word}, which satisfies $\varphi$. For MITL semantics see \cite[Sec. II, p. 2]{alex_2016_acc}).
\end{definition}

\subsection{Problem Statement}

\noindent The problem considered in this paper is stated as follows:

\begin{problem} \label{problem}
Consider a robot governed by dynamics \eqref{eq:kinematics}-\eqref{eq:unc_dynamics}, covered by the ball $\mathcal{B}(\chi(t), \mathfrak{r})$, operating in the workspace $\mathcal{W} \subseteq \mathbb{R}^{n}$. The workspace contains the RoI $\mathcal{R}_m$, $m \in \mathcal{M}$ modeled also by balls. Given a task specification formula $\varphi$ expressed in MITL over the set of atomic propositions $\Sigma$ and labeling functions $L$ as in \eqref{eq:label_function_L}. Then, for every $d \in \mathcal{D}$, design a feedback control law $u = \kappa(\chi, v) \in \mathcal{U}$ such that the robot trajectory in the workspace fulfills the MITL specification $\varphi$, i.e., $\chi(t) \models \varphi$, $\forall t \ge 0$,  according to Definition \ref{def:unique_timed_word}. Moreover, the robot is required to remain in the workspace for all times.
\end{problem}

\begin{remark}
Note that Problem \ref{problem} constitutes a general time-constrained navigation problem due to the fact that the dynamics \eqref{eq:kinematics}-\eqref{eq:unc_dynamics} arise in many robotic applications. Furthermore, the rich expressiveness of MITL in both reachability and safety specifications can be exploited.
\end{remark}

\section{Main Results} \label{sec:main_results}

In this section, the aforementioned control design problem is addressed by taking the following steps:
\begin{enumerate}
\item For navigating the robot between RoI, we propose a robust NMPC feedback law that has two components: an on-line control law which is the outcome of a Finite Horizon Optimal Control Problem (FHOCP) solved at each sampling time; a state feedback law whose gain is designed off-line and guarantees that the trajectory of the closed loop system remains in a hyper-tube for all times. The time duration for navigating between RoI is explicitly provided. (Section \ref{sec:feedback_control_design})
\item Then, the dynamics \eqref{eq:kinematics}-\eqref{eq:unc_dynamics} are abstracted into a WTS, exploiting the fact that the timed runs in the WTS project onto uniquely associated trajectories according to Definition \ref{def:unique_timed_word}. (Section \ref{sec:discrete_system_abstraction})
\item By invoking ideas from our previous works \cite{alex_2016_acc, alex_automatica_2018}, a controller synthesis procedure that gives a sequence of control laws that serve as solution to Problem \ref{problem} is performed. (Section \ref{sec:control_synthesis}).
\end{enumerate}

\subsection{Feedback Control Design} \label{sec:feedback_control_design}

Consider a robot with dynamics \eqref{eq:kinematics}-\eqref{eq:unc_dynamics} occupying a RoI $\mathcal{R}_{\scriptscriptstyle \rm s} \in \mathcal{R}$ at time $\mathfrak{t}_{\scriptscriptstyle \rm s} \ge 0$. The feedback control law needs to guarantee that the robot is navigated towards a desired RoI $\mathcal{R}_{\scriptscriptstyle \rm d} \in \mathcal{R}$, $\mathcal{R}_{\scriptscriptstyle \rm s} \neq \mathcal{R}_{\scriptscriptstyle \rm d}$ without intersecting with any other RoI, due to the fact that safety specifications are required. Denote by $\chi_{\scriptscriptstyle \rm d} \in \mathcal{R}_{\scriptscriptstyle \rm d}$ the center of the desired RoI $\mathcal{R}_{\scriptscriptstyle \rm d}$. Define the error vector $e \coloneqq \chi - \chi_{\scriptscriptstyle \rm d} \in \mathbb{R}^{n}$. The \emph{uncertain error kinematics/dynamics} are given by:
\begin{subequations}
	\begin{align}
	\dot{e} & = v, \label{eq:unc_error_kin} \\
	\dot{v} & = f(e+\chi_{\scriptscriptstyle \rm d}, v, u)+d, \label{eq:unc_error_dyn}
	\end{align}
\end{subequations}
and the corresponding \emph{nominal error kinematics/dynamics} by:
\vspace{-3mm}
\begin{subequations}
	\begin{align}
	\dot{\overline{e}} & = \overline{v}, \label{eq:nom_error_kin} \\
	\dot{\overline{v}} & = f(\overline{e}+\chi_{\scriptscriptstyle \rm d}, \overline{v}, \overline{u}). \label{eq:nom_error_dyn}
	\end{align}
\end{subequations}
By recalling that $\mathcal{B}(\chi(t), \mathfrak{r})$ stands for the volume of the robot at time $t$, define the set that captures the state constraints by:
\begin{align*}
\mathcal{X} & \coloneqq \{\chi(t) \in \mathbb{R}^n : \mathcal{B}(\chi(t), \mathfrak{r}) \subsetneq \mathcal{W}, \mathcal{B}(\chi(t), \mathfrak{r}) \cap \{\mathcal{R} \backslash \{\mathcal{R}_{\scriptscriptstyle \rm s}, \mathcal{R}_{\scriptscriptstyle \rm d} \} \} = \emptyset\}.
\end{align*}
The two constraints refer to the fact that the robot needs to remain in the workspace for all times and the fact that it should not intersect with any other RoI except from $\mathcal{R}_{\scriptscriptstyle \rm s}$, $\mathcal{R}_{\scriptscriptstyle \rm d}$. In order to translate the aforementioned constraints for the error state $e$, define the set $\mathcal{E} \coloneqq \{e \in \mathbb{R}^n : e \in \mathcal{X} \oplus (-\chi_{\scriptscriptstyle \rm d})\}$, where $\oplus$ is the Minkowski addition operator given in Section \ref{sec:notation_preliminaries}. Under this modification, by using basic properties of Minkowski operator $\oplus$, it is guaranteed that $\chi \in \mathcal{X} \Leftrightarrow e \in \mathcal{E}$. Consider the \emph{feedback control law:}
\begin{align} \label{eq:control_law_u}
u \coloneqq \overline{u}(\overline{e}, \overline{v}) + \kappa(e, v, \overline{e}, \overline{v}),
\end{align}
which consists of a nominal control action $\overline{u}(\overline{e}, \overline{v}) \in \mathcal{U}$ and a state feedback law $\kappa : \mathbb{R}^{n} \times \mathbb{R}^{n} \times \mathbb{R}^{n} \times \mathbb{R}^{n} \to \mathbb{R}^{n}$. The control action $\overline{u}(\overline{e}, \overline{v})$ will be the outcome of a nominal FHOCP which is solved on-line at each sampling time. The feedback law $\kappa(\cdot)$ is tuned off-line and it is used to guarantee that the real states $e$, $v$ remain in a bounded hyper-tube centered along the nominal states $\overline{e}$, $\overline{v}$. Define by $\widetilde{e} \coloneqq e - \overline{e} \in \mathbb{R}^{n}$, $\widetilde{v} \coloneqq v - \overline{v} \in \mathbb{R}^n$ the deviation between the real states of the uncertain system \eqref{eq:unc_error_kin}-\eqref{eq:unc_error_dyn} and the states of the nominal system \eqref{eq:nom_error_kin}-\eqref{eq:nom_error_dyn} with $\widetilde{e}(0) = \widetilde{v}(0) = 0$. It will be proved hereafter that the states $\widetilde{e}$, $\widetilde{v}$ are invariant in a compact set whose volume depends of the bounds of the derivatives of $f$ as well as the bound $\widetilde{d}$. The dynamics of the states $\widetilde{e}$, $\widetilde{v}$ are written as:
\begin{subequations}
	\begin{align}
	\dot{\widetilde{e}} & = \dot{e} - \dot{\overline{e}} = v - \overline{v} = \widetilde{v}, \label{eq:tilde_e} \\
	\dot{\widetilde{v}} & = \dot{v} - \dot{\overline{v}} = f(e+\chi_{\scriptscriptstyle \rm d}, v, u) - f(\overline{e}+\chi_{\scriptscriptstyle \rm d}, \overline{v}, \overline{u}) + d \notag \\
	& = f(e+\chi_{\scriptscriptstyle \rm d}, v, u) - f(\overline{e}+\chi_{\scriptscriptstyle \rm d}, \overline{v}, u) f(\overline{e}+\chi_{\scriptscriptstyle \rm d}, \overline{v}, u) - f(\overline{e}+\chi_{\scriptscriptstyle \rm d}, \overline{v}, \overline{u}) + d \notag \\
	& = \lambda(e, \overline{e}, v, \overline{v},u,\overline{u}) + f(\overline{e}+\chi_{\scriptscriptstyle \rm d}, \overline{v}, u) - f(\overline{e}+\chi_{\scriptscriptstyle \rm d}, \overline{v}, \overline{u}) + d. \label{eq:tilde_v}
	\end{align}
\end{subequations}
where the function $\lambda$ is defined by $\lambda(e, \overline{e}, v, \overline{v},u,\overline{u}) \coloneqq f(e+\chi_{\scriptscriptstyle \rm d}, v, u)  - f(\overline{e}+\chi_{\scriptscriptstyle \rm d}, \overline{v}, u)$ and is upper bounded by:
\begin{align}
\|\lambda(\cdot)\|_{\scriptscriptstyle 2} & \le \|(e+\chi_{\scriptscriptstyle \rm d}, v, u) - f(\overline{e}+\chi_{\scriptscriptstyle \rm d}, v, u)\|_{\scriptscriptstyle 2} + \|f(\overline{e}+\chi_{\scriptscriptstyle \rm d}, v, u) - f(\overline{e}+\chi_{\scriptscriptstyle \rm d}, \overline{v}, u)\|_{\scriptscriptstyle 2} \notag \\
& \le L_1 \|\widetilde{e}\|_{\scriptscriptstyle 2} + L_2 \|\widetilde{v}\|_{\scriptscriptstyle 2} \notag \\
& \le L (\|\widetilde{e}\|_{\scriptscriptstyle 2}+\|\widetilde{v}\|_{\scriptscriptstyle 2}). \notag
\end{align}
The constants $L_1$, $L_2$ stand for the Lipschitz constants of function $f$ with respect to variables $\chi$ and $v$, respectively, and $L \coloneqq \max\{L_1, L_2\}$.
\begin{lemma} \label{lemma:rci_set}
	The state feedback law designed as:
	\begin{align} \label{eq:kappa_law}
	\kappa(e, \overline{e}, v, \overline{v}) = -k(e-\overline{e})-k(v-\overline{v}),
	\end{align}
	where $k > 0$ is chosen such that:
	\begin{align} \label{eq:k_rho}
	k > \tfrac{1}{\underline{J}} \left[(1+2\rho) L + \tfrac{5}{4}\right], \ \rho > \tfrac{L}{2},
	\end{align}
	renders the sets:
	\begin{subequations}
		\begin{align}
		\Omega_1  \coloneqq  \Big\{\widetilde{e}(t) \in \mathbb{R}^n : \|\widetilde{e}(t)\|_{\scriptscriptstyle 2} & \le  \tfrac{\widetilde{d}}{\sqrt{\min\{\alpha_1,\alpha_2}\}}, \forall t \ge 0 \Big\},\label{eq:omega_1} \\
		\Omega_2  \coloneqq  \Big\{\widetilde{v}(t) \in \mathbb{R}^n : \|\widetilde{v}(t)\|_{\scriptscriptstyle 2} & \le \tfrac{2 \widetilde{d}}{\sqrt{\min\{\alpha_1,\alpha_2}\}}, \forall t \ge 0 \Big\}, \label{eq:omega_2}
		\end{align}
	\end{subequations}
	RCI sets for the error dynamics \eqref{eq:tilde_e}, \eqref{eq:tilde_v}, according to Definition \ref{def:RPI_set}. The constants $\alpha_1$, $\alpha_2 > 0$ are defined by:
	\begin{align} \label{eq:cosntants_alpha}
	\alpha_1 \coloneqq 1- \tfrac{L}{2 \rho}, \alpha_2 \coloneqq k \underline{J}-(1 + 2\rho) L -\tfrac{5}{4}.
	\end{align}
\end{lemma}
\begin{proof}
	A backstepping control design technique originally proposed in \cite{krstic1995nonlinear} will be adopted. The state $\widetilde{v}$ can be seen as a virtual input to be designed for system \eqref{eq:tilde_e} such that the function $\mathfrak{L}_1(\overline{e}) = \tfrac{1}{2} \|\widetilde{e}\|^2_{2}$ is always decreasing. The time derivative of $\mathfrak{L}_1$ along the solutions of system \eqref{eq:tilde_e}-\eqref{eq:tilde_v} is given by $\dot{\mathfrak{L}}_1(\widetilde{e}) = \widetilde{e}^\top \widetilde{v}$. Thus, by designing $\widetilde{v} = - \widetilde{e}$, it yields that $\dot{\mathfrak{L}}_1(\widetilde{e}) = - \|\widetilde{e}\|^2_{2} < 0$. Define the backstepping auxiliary errors $\zeta_1$, $\zeta_2 \in \mathbb{R}^n$ by $\zeta_1 \coloneqq \widetilde{e}$ and $\zeta_2 \coloneqq \widetilde{v}+\widetilde{e}$. Then, the auxiliary error dynamics are written as:
	\begin{subequations}
		\begin{align}
		\dot{\zeta}_1 & = -\zeta_1+\zeta_2, \label{eq:zeta_1} \\
		\dot{\zeta}_2 & = -\zeta_1+\zeta_2 + \lambda(\cdot) \notag \\ 
		&\hspace{19mm} + f(\overline{e}+\chi_{\scriptscriptstyle \rm d}, \overline{v}, u) - f(\overline{e}+\chi_{\scriptscriptstyle \rm d}, \overline{v}, \overline{u}) + d, \label{eq:zeta_2}
		\end{align}
	\end{subequations}
	with:
	\begin{align*}
	\|\lambda(\cdot)\|_{\scriptscriptstyle 2} & \le L(\|\widetilde{e}\|_{\scriptscriptstyle 2}+\|\widetilde{v}\|_{\scriptscriptstyle 2}) \\
	& = L(\|\zeta_1\|_{\scriptscriptstyle 2} +\|\zeta_1-\zeta_2\|_{\scriptscriptstyle 2}) \\
	& \le 2 L \|\zeta_1\|_{\scriptscriptstyle 2} + L \|\zeta_2\|_{\scriptscriptstyle 2},
	\end{align*}
	 and $\zeta_1(0) = \zeta_2(0) = 0$. Define the stack vector $\zeta \coloneqq [\zeta_1^\top, \zeta_2^\top]^\top \in \mathbb{R}^{2n}$ and consider the candidate Lyapunov function $\mathfrak{L}(\zeta) = \tfrac{1}{2} \|\zeta\|_2^2$ with $\mathfrak{L}(0) = 0$. The time derivative of $\mathfrak{L}$ along the trajectories of system \eqref{eq:zeta_1}-\eqref{eq:zeta_2} is given by:
	\begin{align*}
	\dot{\mathfrak{L}}(\zeta) & = \zeta^\top \dot{\zeta} = \zeta_1^\top \dot{\zeta}_1 + \zeta_2^\top \dot{\zeta}_2 \\
	& \le - \|\zeta_1\|_2^2+(L+1)\|\zeta_2\|_2^2 + 2 L \|\zeta_1\|_{\scriptscriptstyle 2}  \|\zeta_2\|_{\scriptscriptstyle 2} + \zeta_2^\top d \\
	& \hspace{43mm} +\zeta_2^\top \left[f(\overline{e}+\chi_{\scriptscriptstyle \rm d}, \overline{v}, u) - f(\overline{e}+\chi_{\scriptscriptstyle \rm d}, \overline{v}, \overline{u})\right].
	\end{align*} 
	By using Lemma \ref{lemma:basic_ineq} for $B = I_2$ we have: 
	\begin{align*}
	\|\zeta_1\|_{\scriptscriptstyle 2} \|\zeta_2\|_{\scriptscriptstyle 2} & \le \tfrac{\|\zeta_1\|_2^2}{4 \rho}+\rho\|\zeta_2\|_2^2, \\ 
	\zeta_2^\top d & \le \tfrac{\|\zeta_2\|_2^2}{4}+\|d\|_2^2,
	\end{align*}
	with $\rho$ as given in \eqref{eq:k_rho}. Then, it holds that:
	\begin{align*}
	\dot{\mathfrak{L}}(\zeta) \le & - \left(1- \tfrac{L}{2 \rho} \right) \|\zeta_1\|_2^2 + \left[(1 + 2\rho) L +\tfrac{5}{4}\right] \|\zeta_2\|_2^2 \notag \\
	&\hspace{29mm} +\zeta_2^\top \left[f(\overline{e}+\chi_{\scriptscriptstyle \rm d}, \overline{v}, u) - f(\overline{e}+\chi_{\scriptscriptstyle \rm d}, \overline{v}, \overline{u})\right] +\widetilde{d}^2.
	\end{align*}
	According to Proposition \ref{pro:mvt}, and due to the fact that the set $\mathcal{U}$ is convex, i.e., $\mathcal{C}(u, \overline{u}) \subseteq \mathcal{U}$, there exist constant vectors $\varpi_1$, $\dots$, $\varpi_n \in \mathcal{C}(u,\overline{u})$ such that:
	\begin{align*}
	f(\overline{e}+\chi_{\scriptscriptstyle \rm d}, \overline{v}, u) - f(\overline{e}+\chi_{\scriptscriptstyle \rm d}, \overline{v}, \overline{u}) & = \sum_{i = 1}^{n} \sum_{j=1}^{n} \mathfrak{e}_n^i (\mathfrak{e}_n^j)^\top \tfrac{\partial f_{i}(\overline{e}+\chi_{\scriptscriptstyle \rm d}, \overline{v},\varpi_k)}{\partial u_j} (u-\overline{u}) \\ 
	& = J(\overline{e}+\chi_{\scriptscriptstyle \rm d}, \overline{v},\cdot) (u-\overline{u}).	
	\end{align*}
    Then, we get:
	\begin{align*}
	& \dot{\mathfrak{L}}(\zeta) \le - \left(1- \tfrac{L}{2 \rho} \right) \|\zeta_1\|_2^2 + \left[(1 + 2\rho) L +\tfrac{5}{4}\right] \|\zeta_2\|_2^2 \notag \\
	&\hspace{41mm}+\zeta_2^\top J(\overline{e}+\chi_{\scriptscriptstyle \rm d}, \overline{v},\cdot) (u-\overline{u}) +\widetilde{d}^2.
	\end{align*}
	By designing the feedback control law as $u-\overline{u}$ $= \kappa(e,v,\overline{e},\overline{v})$ $= -k \zeta_2 = -k\widetilde{e}-k\widetilde{v}$ $= -k(e-\overline{e})-k(v-\overline{v})$, which is the same as in \eqref{eq:kappa_law} and compatible with \eqref{eq:control_law_u}, we get:
	\begin{align*}
	\dot{\mathfrak{L}}(\zeta) & \le - \left(1- \tfrac{L}{2 \rho} \right) \|\zeta_1\|_2^2 + \left[(1 + 2\rho) L +\tfrac{5}{4}\right] \|\zeta_2\|_2^2 -k \zeta_2^\top J(\overline{e}+\chi_{\scriptscriptstyle \rm d}, \overline{v},\cdot) \zeta_2 +\widetilde{d}^2 \\
	& \le - \left(1- \tfrac{L}{2 \rho} \right) \|\zeta_1\|_2^2 + \left[(1 + 2\rho) L +\tfrac{5}{4}\right] \|\zeta_2\|_2^2 -k \lambda_{\scriptscriptstyle \min}\left[\tfrac{J(\cdot)+J(\cdot)^\top}{2}\right]\|\zeta_2\|_2^2 +\widetilde{d}^2,
	\end{align*}
	which by using \eqref{eq:J_lower_bound} of Assumption \ref{ass:J_lower_bound} becomes:
	\begin{align*}
	\mathfrak{L}(\zeta) & \le - \alpha_1 \|\zeta_1\|_2^2 - \alpha_2 \|\zeta_2\|_2^2 +\widetilde{d}^2 \\
	& \le -\min\{\alpha_1, \alpha_2 \} \|\zeta\|_2^2 + \widetilde{d}^2,
	\end{align*}
	with $\alpha_1$, $\alpha_2$ given in \eqref{eq:cosntants_alpha}. Thus, $\dot{L}(\zeta) < 0$ when $\|\zeta\|_{\scriptscriptstyle 2} > \tfrac{\widetilde{d}}{\sqrt{\min\{\alpha_1,\alpha_2}\}}$. Taking the latter into consideration and the fact that $\zeta(0) = 0$ we have that $\|\zeta(t)\|_{\scriptscriptstyle 2} \le \tfrac{\widetilde{d}}{\sqrt{\min\{\alpha_1,\alpha_2}\}}$, $\forall t \ge 0$. In order to transform the bounds of the auxiliary states $\zeta_1$, $\zeta_2$ and $\zeta$ into bounds on the states $\widetilde{e}$, $\widetilde{v}$, consider the following inequalities:
	\begin{align*}
	\|\widetilde{e}\|_{\scriptscriptstyle 2} & = \|\zeta_1\|_{\scriptscriptstyle 2} \le \|\zeta\|_{\scriptscriptstyle 2}, \\
	\Rightarrow \|\widetilde{e}(t)\|_{\scriptscriptstyle 2} &\le \tfrac{\widetilde{d}}{\sqrt{\min\{\alpha_1,\alpha_2}\}}, \forall t \ge 0,
	\end{align*}
	and:
	\begin{align*}
	\Big| \|\widetilde{v}\|_{\scriptscriptstyle 2} - \|\widetilde{e}\|_{\scriptscriptstyle 2} \Big| & \le \|\widetilde{v}+\widetilde{e}\|_{\scriptscriptstyle 2}  = \|\zeta_2\|_{\scriptscriptstyle 2} \le \|\zeta\|_{\scriptscriptstyle 2} \\
	\Rightarrow \|\widetilde{v}(t)\|_{\scriptscriptstyle 2} & \le \tfrac{2 \widetilde{d}}{\sqrt{\min\{\alpha_1,\alpha_2}\}}, \forall t \ge 0,
	\end{align*}
	which concludes the proof.
\end{proof}

\begin{remark}
	According to Lemma \ref{lemma:rci_set}, the volume of the hyper-tube which is centered along the nominal trajectories $\overline{e}(t)$, $\overline{v}(t)$ of system \eqref{eq:nom_error_kin}-\eqref{eq:nom_error_dyn}, depends on $\widetilde{d}$, which is the upper bound of the disturbances $d$, and on $L$, $\underline{J}$, which are the bounds of the derivatives of $f$. By tuning the parameters $k$, $\rho$ from \eqref{eq:k_rho} appropriately, the volume of the tube can be adjusted.
\end{remark}

\begin{assumption} \label{ass:feasible_sol}
	It holds that:
	\begin{align*}
	\inf_{\substack{m,m' \in \mathcal{M}, \\ m \neq m'}}\|p_m-p_{m'}\|_{\scriptscriptstyle 2} > 2\mathfrak{r} + \tfrac{2 \widetilde{d}}{\sqrt{\min\{\alpha_1,\alpha_2}\}},
	\end{align*}
	i.e., there is sufficient space between any RoI such that the robot can navigate without intersecting them.
\end{assumption}

\begin{remark}
Assumptions \ref{ass:f_diff}, \ref{ass:stabiliz} are standard assumptions required for the NMPC nominal stability to be guaranteed (see \cite{frank_1998_quasi_infinite}). Assumption \ref{ass:J_lower_bound} is a sufficient controllability condition for nonlinear systems in non-affine form (see \cite{adaptive_non_affine}). Assumption \ref{ass:feasible_sol} is a geometric condition that the workspace should fulfill in order for the robot to be able to navigate between RoI and meet the time constraints imposed by the desired formula~$\varphi$.
\end{remark}

Consider a sequence of sampling times $\{t_k\}$, $k \in \mathbb{N}$, with a constant sampling period $0 < h < T$, where $T$ is a \emph{finite prediction horizon} such that $t_{k+1} \coloneqq t_{k} + h$, $\forall k \in \mathbb{N}$. At every discrete sample time, a FHOCP is solved as follows:
\begin{subequations}
\begin{align}
&\hspace{-7mm}\min\limits_{\overline{u}(\cdot)} \left\{  \|\overline{\xi}(t_k+T)\|^2_{\scriptscriptstyle P} + \int_{t_k}^{t_k+T}\Big[ \|\overline{\xi}(\mathfrak{s})\|^2_{\scriptscriptstyle Q} +\|\overline{u}(\mathfrak{s})\|^2_{\scriptscriptstyle R} \Big] d\mathfrak{s} \right\} \label{eq:mpc_cost_function} \\
&\hspace{-6mm}\text{subject to:} \notag \\
&\hspace{-3mm} \dot{\overline{\xi}}(\mathfrak{s}) = g(\overline{\xi}(\mathfrak{s}), \overline{u}(\mathfrak{s})), \ \ \overline{\xi}(t_k) = \xi(t_k), \label{eq:diff_mpc} \\
&\hspace{-3mm} \overline{\xi}(\mathfrak{s}) \in \overline{\mathcal{E}} \times \overline{\mathcal{V}}, \ \ \overline{u}(\mathfrak{s}) \in \overline{\mathcal{U}},  \ \ \forall \mathfrak{s} \in [t_k,t_k+T], \label{eq:mpc_constrained_set} \\
&\hspace{-3mm} \overline{\xi}(t_k+T)\in \mathcal{F}, \label{eq:mpc_terminal_set}
\end{align}
\end{subequations}
where $\xi \coloneqq [e^\top, v^\top]^\top \in \mathbb{R}^{2n}$, $g(\xi,u)$ $\coloneqq \begin{bmatrix} v \\ f(e+\chi_{\scriptscriptstyle \rm d}, v, u) \end{bmatrix}$; $Q$, $P \in \mathbb{R}^{2n \times 2n}$ and $R \in \mathbb{R}^{n \times n}$ are positive definite gain matrices to be appropriately tuned. We will explain hereafter the sets $\overline{\mathcal{E}}$, $\overline{\mathcal{V}}$, $\overline{\mathcal{U}}$ and $\mathcal{F}$.

\begin{figure*}
	\begin{tikzpicture}[scale = 0.90] 
	% convert formulas phi_i to buchis A_i
	\draw[blue!70, line width=.04cm] (-16.8, 6.0) rectangle +(2.7, 0.9);
	\node at (-15.45, 6.45) {$\text{MITL2TBA}$};
	
	\draw[-latex, draw=black, line width = 1.0] (-15.5,6.0) -- (-15.5,4.6);
	\draw[-latex, draw=black, line width = 1.0] (-15.5,7.6) -- (-15.5,6.9);
	
	\node at (-15.5, 8.00) {$\varphi$};
	\node at (-15.0, 5.35) {$\mathcal{A}$};
	\node at (-15.5, 4.45) {$\otimes$};
	
	%\node at (-15.5, 3.95) {$\vdots$};
	
	% create \tilde T_i
	\draw[-latex, draw=black, line width = 1.0] (-15.20,4.45) -- (-14.0,4.45);
	
	\node at (-13.6, 4.47) {$\widetilde{\mathcal{T}}$};
	
	% create T_i
	\draw[-latex, draw=black, line width = 1.0] (-16.60,4.45) -- (-15.7,4.45);
	\node at (-17.2, 4.47) {$\mathcal{T}$};
	
	% create box algorithm
	\draw[-latex, draw=black, line width = 1.0] (-13.20,4.45) -- (-12.5,4.45);
	\draw[red!70, line width=.04cm] (-12.5, 4.15) rectangle +(2.20, 0.70);
	\node at (-11.32, 4.50) {$\text{synthesis}$};
	
	% create runs
	\draw[-latex, draw=black, line width = 1.0] (-10.25,4.45) -- (-9.75,4.45);
	
	\node at (-9.35, 4.47) {$\widetilde{r}$};
	
	% create abstractions box
	\draw[-latex, draw=black, line width = 1.0] (-18.45,4.47) -- (-17.90,4.47);
	\draw[orange!70, line width=.04cm] (-21.0, 4.20) rectangle +(2.5, 0.7);
	\node at (-19.70, 4.60) {$\text{abstraction}$};
	
	\draw[-latex, draw=black, line width = 1.0] (-21.70,4.47) -- (-21.00,4.47);
	
	% create the dynamics
	\node at (-24.0, 4.50) {$\hspace{-5mm}\displaystyle \dot{\chi}= v$};
	\node at (-24.0, 4.00) {$\hspace{13mm} \dot{v} = f(x,v,u)+d$};
	
	% create control inputs v_1,...,v_N
	\draw [black, line width = 0.030cm] (-9.35, 4.80) -- (-9.35, 8.50);
	\draw [black, line width = 0.030cm] (-9.35, 8.50) -- (-18.90, 8.50);
	\draw [black, line width = 0.030cm] (-20.20, 8.50) -- (-24.00, 8.50);
	\draw[-latex, draw=black, line width = 1.0] (-24.00, 8.50) -- (-24.00, 5.1);
	
	\draw[green!70, line width=.04cm] (-20.2, 7.95) rectangle +(1.3, 1.0);
	\node at (-19.5, 8.40) {$u(x,v)$};
	\end{tikzpicture}
	\caption{A graphical illustration of the combined abstraction and controller synthesis framework.}
	\label{fig:solution_scheme}
	\vspace{-2mm}
\end{figure*}
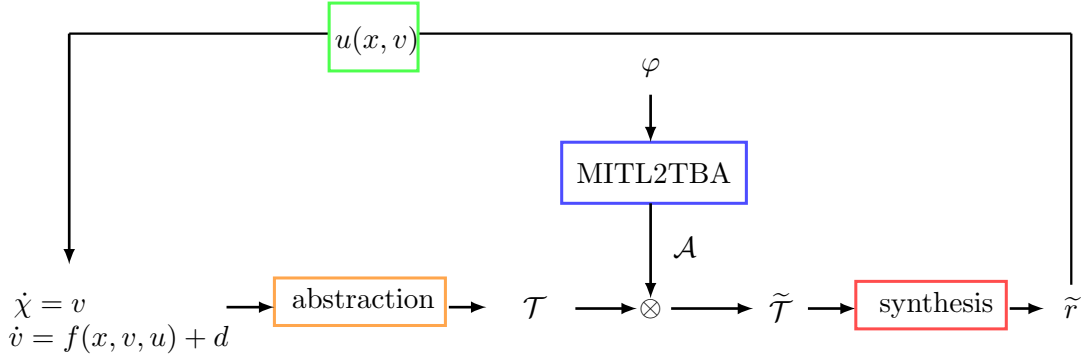

In order to guarantee that while the FHOCP \eqref{eq:mpc_cost_function}-\eqref{eq:mpc_terminal_set} is solved for the nominal dynamics \eqref{eq:nom_error_kin}-\eqref{eq:nom_error_dyn}, the real states $e$, $v$ and control inputs $u$ satisfy the corresponding state $\mathcal{E}$, $\mathcal{V}$ and input constraints $\mathcal{U}$, respectively, the following modification is performed:
\begin{align*}
\overline{\mathcal{E}} \coloneqq \mathcal{E} \ominus \Omega_1, \ \ \overline{\mathcal{V}} \coloneqq \mathcal{V} \ominus \Omega_2, \ \ \overline{\mathcal{U}} \coloneqq \mathcal{U} \ominus \left[ \Lambda \circ \overline{\Omega} \right],
\end{align*}
with $\Lambda \coloneqq {\rm diag}\{-k I_n, -k I_n\} \in \mathbb{R}^{2n \times 2n}$, $\overline{\Omega} \coloneqq \Omega_1 \times \Omega_2$, the operators $\ominus$, $\circ$ as defined in Section \ref{sec:notation_preliminaries}, and $\Omega_1$, $\Omega_2$ as given in \eqref{eq:omega_1}, \eqref{eq:omega_2}, respectively. Intuitively, the sets $\mathcal{E}$, $\mathcal{V}$ and $\mathcal{U}$ are tightened accordingly, in order to guarantee that while the nominal states $\overline{e}$, $\overline{v}$ and the nominal control input $\overline{u}$ are calculated, the corresponding real states $e$, $v$ and real control input $u$ satisfy the state and input constraints $\mathcal{E}$, $\mathcal{V}$ and $\mathcal{U}$, respectively. This constitutes a standard constraints set modification technique adopted in tube-based NMPC frameworks (for more details see \cite{yu_2013_tube}, \cite{alex_IJRNC_2018}). Define the \emph{terminal set} by:
\begin{align} \label{eq:terminal_set_F}
\mathcal{F} \coloneqq \big\{\overline{\xi} \in \overline{\mathcal{E}} \times \overline{\mathcal{V}} : \|\overline{\xi}\|_{\scriptscriptstyle P} \le \epsilon \big\}, \ \  \epsilon > 0,
\end{align}
which is used to enforce the stability of the system \cite{frank_1998_quasi_infinite}. In particular, supposing that Assumption \ref{ass:f_diff}, \ref{ass:stabiliz} hold, it can be proven that (see \cite[Lemma 1, p. 4]{frank_1998_quasi_infinite}), for a sufficiently small $\epsilon$, there exists a \emph{local controller} $u_{\scriptscriptstyle \rm loc} \coloneqq K \overline{\xi} \in \overline{\mathcal{U}}$, $K > 0$ which guarantees that:
\begin{align*}
\tfrac{d}{dt}\left(\|\overline{\xi}\|^2_{\scriptscriptstyle P}\right) \le -\|\overline{\xi}\|^2_{\scriptscriptstyle \widetilde{Q}}, \forall \overline{\xi} \in \mathcal{F},
\end{align*}
with $\widetilde{Q} \coloneqq Q+K^\top R K > 0$. The following theorem guarantees the robust robot navigation from RoI $\mathcal{R}_{\scriptscriptstyle \rm s}$ to RoI $\mathcal{R}_{\scriptscriptstyle \rm d}$ without intersecting any other RoI and always remaining in the workspace $\mathcal{W}$.

\begin{theorem} \label{theorem_main}
	Suppose that Assumptions \ref{ass:f_diff}-\ref{ass:feasible_sol} hold. Let $\mathfrak{t}_{\scriptscriptstyle \rm s} \ge 0$ be the time at which the robot occupies RoI $\mathcal{R}_{\scriptscriptstyle \rm s}$ and $\chi_{\scriptscriptstyle \rm d}$ be the center of a desired RoI $\mathcal{R}_{\scriptscriptstyle \rm d}$. Suppose also that the FHOCP \eqref{eq:mpc_cost_function}-\eqref{eq:mpc_terminal_set} is feasible at time $\mathfrak{t}_{\scriptscriptstyle \rm s}$. Then, the feedback control law \eqref{eq:control_law_u} applied to the system \eqref{eq:unc_error_kin}-\eqref{eq:unc_error_dyn} guarantees that there exists a finite time $\mathfrak{t}_{\scriptscriptstyle \rm d} > \mathfrak{t}_{\scriptscriptstyle \rm s}$ such that $\forall t \ge \mathfrak{t}_{\scriptscriptstyle \rm d}$ it holds that:
	\begin{subequations}
		\begin{align}
		\hspace{0mm} \|\chi(t)-\chi_{\scriptstyle \rm d}\|_{\scriptscriptstyle 2} & \le \tfrac{\epsilon}{\sqrt{\lambda_{\scriptscriptstyle \min}(P)}} + \tfrac{\widetilde{d}}{\sqrt{\min\{\alpha_1, \alpha_2\}}}, \label{eq:theom_ineq_1} \\
		\hspace{0mm} \|v(t)\|_{\scriptscriptstyle 2} & \le \tfrac{\epsilon}{\sqrt{\lambda_{\scriptscriptstyle \min}(P)}} + \tfrac{2 \widetilde{d}}{\sqrt{\min\{\alpha_1, \alpha_2\}}}. \label{eq:theom_ineq_2}
		\end{align}
	\end{subequations}
\end{theorem}
\begin{proof}
	The proof of the theorem consists of two parts:
	
	\noindent \textbf{Feasibility Analysis}: It can be shown that recursive feasibility is established and it implies subsequent feasibility. The proof of this part is similar to the feasibility proof of \cite[Theorem 2, Sec. 4, p. 12]{alex_IJRNC_2018}, and it is omitted here due to space constraints.
	
	\noindent \textbf{Convergence Analysis}: Recall that $e = \chi-\chi_{\scriptscriptstyle \rm d}$, $\widetilde{e} = e-\overline{e}$ and $\widetilde{v} = v-\overline{v}$. Then, we get:
	\begin{align*}
	\|\chi(t)-\chi_{\scriptstyle \rm d}\|_{\scriptscriptstyle 2} & \le \|\overline{e}(t)\|_{\scriptscriptstyle 2} + \|\widetilde{e}(t)\|_{\scriptscriptstyle 2}, \\ \|v(t)\|_{\scriptscriptstyle 2} & \le \|\overline{v}(t)\|_{\scriptscriptstyle 2} + \|\widetilde{v}(t)\|_{\scriptscriptstyle 2},
	\end{align*}
	which by using
	the fact that $\|\overline{e}\|_{\scriptscriptstyle 2}$, $\|\overline{v}\|_{\scriptscriptstyle 2}$ $\le \|\overline{\xi}\|_{\scriptscriptstyle 2}$ it becomes:
	\begin{align*}
	\|\chi(t)-\chi_{\scriptstyle \rm d}\|_{\scriptscriptstyle 2} & \le \|\overline{\xi}(t)\|_{\scriptscriptstyle 2} + \|\widetilde{e}(t)\|_{\scriptscriptstyle 2}, \\ \|v(t)\|_{\scriptscriptstyle 2} & \le \|\overline{\xi}(t)\|_{\scriptscriptstyle 2} + \|\widetilde{v}(t)\|_{\scriptscriptstyle 2}.
	\end{align*}
	Moreover, by using the bounds from \eqref{eq:omega_1}, \eqref{eq:omega_2} the latter inequalities become:
	\vspace{-4mm}
	\begin{subequations}
		\begin{align}
		\|\chi(t)-\chi_{\scriptstyle \rm d}\|_{\scriptscriptstyle 2} & \le \|\overline{\xi}(t)\|_{\scriptscriptstyle 2} +  \tfrac{ \widetilde{d}}{\sqrt{\min\{\alpha_1, \alpha_2\}}}, \label{eq:conv_1}\\
		\|v(t)\|_{\scriptscriptstyle 2} & \le \|\overline{\xi}(t)\|_{\scriptscriptstyle 2} +  \tfrac{2 \widetilde{d}}{\sqrt{\min\{\alpha_1, \alpha_2\}}}, \forall t \ge \mathfrak{t}_{\scriptscriptstyle \rm s}. \label{eq:conv_2}
		\end{align}
	\end{subequations}
	The nominal state $\overline{\xi}$ is controlled by the nominal control action $\overline{u} \in \overline{\mathcal{U}}$ which is the outcome of the solution to the  FHOCP \eqref{eq:mpc_cost_function}-\eqref{eq:mpc_terminal_set} for the nominal dynamics \eqref{eq:nom_error_kin}-\eqref{eq:nom_error_dyn}. Hence, by invoking nominal NMPC stability results found in \cite{frank_1998_quasi_infinite, alex_med_2017}, the state $\overline{\xi}(t)$ is driven to terminal set $\mathcal{F}$, given in \eqref{eq:terminal_set_F}, in finite time, and it remains there for all times. Thus, there exist a finite time $\mathfrak{t}_{\scriptscriptstyle \rm d} > \mathfrak{t}_{\scriptscriptstyle \rm s}$ such that $\overline{\xi}(t) \in \mathcal{F}$, $\forall t \ge \mathfrak{t}_{\scriptscriptstyle \rm d}$. From \eqref{eq:terminal_set_F}, the latter implies that: $$\|\overline{\xi}(t)\|_{\scriptscriptstyle P} \le \epsilon, \forall t \ge \mathfrak{t}_{\scriptstyle \rm d} \Rightarrow \|\overline{\xi}(t)\|_{\scriptscriptstyle 2} \le \tfrac{\epsilon}{\sqrt{\lambda_{\scriptscriptstyle \min}(P)}}, \forall t \ge \mathfrak{t}_{\scriptstyle \rm d}.$$
	The latter implication combined by \eqref{eq:conv_1}-\eqref{eq:conv_2} leads to the conclusion of the proof.
\end{proof}

Theorem \ref{theorem_main} implies that the robot with dynamics as in \eqref{eq:kinematics}-\eqref{eq:unc_dynamics}, starting at time $\mathfrak{t}_{\scriptscriptstyle \rm s}$ in RoI $\mathcal{R}_{\scriptscriptstyle \rm s}$, is driven by the controller \eqref{eq:control_law_u} towards a desired RoI  $\mathcal{R}_{\scriptscriptstyle \rm d}$, while all constraints imposed to the system are satisfied, i.e., the robot does not intersects with other RoI and always remains in the workspace $\mathcal{W}$. Moreover, by observing \eqref{eq:theom_ineq_1} it holds that at time $\mathfrak{t}_{\scriptscriptstyle \rm d}$ the error $\|\chi(t)-\chi_{\scriptscriptstyle \rm d}\|_2$ has reached the steady-state, i.e., the robot has been navigated to the desired RoI $\mathcal{R}_{\scriptscriptstyle \rm d}$ at time $\mathfrak{t}_{\scriptscriptstyle \rm d}$. Recalling \eqref{eq:desired_times_T} and taking into consideration the aforementioned discussion, $\mathfrak{t}_{\scriptscriptstyle \rm d}$ models the time that the robot needs to be driven from $\mathcal{R}_{\scriptscriptstyle \rm s}$ to $\mathcal{R}_{\scriptscriptstyle \rm d}$, i.e., $\mathfrak{t}_{\scriptscriptstyle \rm d} = \mathfrak{t}(\mathcal{R}_{\scriptscriptstyle \rm s}, \mathcal{R}_{\scriptscriptstyle \rm d})$, and it can be computed by Algorithm $1$. Intuitively, as time evolves, the norm of the states of the robot robot is being monitored at each sampling time $t_k$ until the inequalities of line $7$ of Algorithm $1$ are satisfied, i.e., when the trajectory of the robot has reached the steady-state. When they are satisfied, the robot is within RoI $\mathcal{R}_{\scriptscriptstyle \rm d}$ and the time constant $\mathfrak{t}_{\scriptscriptstyle \rm d} = \mathfrak{t}(\mathcal{R}_{\scriptscriptstyle \rm s}, \mathcal{R}_{\scriptscriptstyle \rm d})$ has been computed.

\begin{algorithm}[t!]
	\caption{Computation of $\mathfrak{t}_{\scriptscriptstyle \rm d} \coloneqq \mathfrak{t}(\mathcal{R}_{\scriptscriptstyle \rm s}, \mathcal{R}_{\scriptscriptstyle \rm d})$}
	\begin{algorithmic}[1]
		\STATE \textbf{Input}: $\mathfrak{t}_{\scriptscriptstyle \rm s}$, $\overline{\chi}(t_k)$, $\overline{v}(t_k)$, $k \in \mathbb{N}$;
		\STATE \textbf{Output}:  $\mathfrak{t}_{\scriptscriptstyle \rm d}$;
		\STATE $t_k \leftarrow \mathfrak{t}_{\scriptscriptstyle \rm s}$;
		\STATE $\rm flag \leftarrow 1$
		\WHILE {$\rm flag = 1$}
		\STATE measure $\overline{\chi}(t_k)$, $\overline{v}(t_k)$;
		\IF {$\|\overline{\chi}(t_k)-\chi_{\scriptscriptstyle d}\|_2 \le \tfrac{\epsilon}{\sqrt{\lambda_{\scriptscriptstyle \min}(P)}}$ \textbf{and} $\|\overline{v}(t_k)\|_2 \le \tfrac{\epsilon}{\sqrt{\lambda_{\scriptscriptstyle \min}(P)}}$} 
		\STATE ${\rm flag} \leftarrow 0$; ${\rm break}$; 
		\STATE \textbf{Go to} ``line $12$"  
		\ENDIF
		\STATE $t_k \leftarrow t_k + h$; 
		\ENDWHILE		
		\STATE $\mathfrak{t}_{\scriptscriptstyle \rm d} \leftarrow t_k$;   
	\end{algorithmic} 
\end{algorithm}

\subsection{Discrete System Abstraction} \label{sec:discrete_system_abstraction}

We have provided so far a feedback control law that drives the robot with dynamics as in \eqref{eq:kinematics}-\eqref{eq:unc_dynamics} from RoI $\mathcal{R}_{\scriptscriptstyle \rm s}$ to RoI $\mathcal{R}_{\scriptscriptstyle \rm d}$ within time $\mathfrak{t}(\mathcal{R}_{\scriptscriptstyle \rm s}, \mathcal{R}_{\scriptscriptstyle \rm d})$. The abstraction that captures the dynamics of the robot into a WTS is given through he following definition:

\begin{definition} \label{def:WTS_abstraction}
	The motion of the robot in the workspace $\mathcal{W}$ is modeled by the WTS $\ \mathcal{T}$ $= (S$, $S_{0}$, $\rm Act$, $\longrightarrow$, $\mathfrak{t}$, $\Sigma$, $L)$ where: 
	\begin{itemize}
		\item $S = \mathcal{R} = \bigcup_{m \in \mathcal{M}} \mathcal{R}_m$ is the set of states of the robot that contains all the RoI of the workspace $\mathcal{W}$; 
		\item $S_{0} \subseteq S$ is a set of initial states defined by the robot' s initial position $\chi(0)$ in the workspace;
		\item $\rm Act$ is the set of actions containing the union of all feedback controllers \eqref{eq:control_law_u} which can navigate the robot between  RoI;
		\item $\longrightarrow \subseteq S \times \rm Act \times S$ is the transition relation. We say that $(\mathcal{R}_{\scriptscriptstyle \rm s}, u, \mathcal{R}_{\scriptscriptstyle \rm d}) \in \longrightarrow$, with $\mathcal{R}_{\scriptscriptstyle s}$, $\mathcal{R}_{\scriptscriptstyle s} \in \mathcal{R}$ with $\mathcal{R}_{\scriptscriptstyle s} \neq \mathcal{R}_{\scriptscriptstyle d}$ if there exist feedback control law $u \in Act$ as in \eqref{eq:control_law_u} which can drive the robot from the region $\mathcal{R}_{\scriptscriptstyle \rm s}$ to the region $\mathcal{R}_{\scriptscriptstyle \rm d}$ without intersecting with any other RoI of the workspace; 
		\item $\mathfrak{t}$ and $L$ is the time weight and the labeling function as given in \eqref{eq:label_function_L} and \eqref{eq:desired_times_T}, respectively; $\Sigma$ is the set of atomic propositions imposed by Problem \ref{problem}.
	\end{itemize}
\end{definition}
The aforementioned WTS will allow us to work directly at the discrete level and design a sequence of feedback controllers as in \eqref{eq:control_law_u} that solve Problem $1$. By construction, each timed run produced by the WTS $\mathcal{T}$, as timed run given in Definition \ref{run_of_WTS}, is uniquely associated with the trajectory $\chi(t)$ of the system \eqref{eq:kinematics}-\eqref{eq:unc_dynamics}, as given in Definition \ref{def:unique_timed_word}. Hence, if we find a timed run of $\mathcal{T}$ satisfying the given MITL formula $\varphi$, we also find a desired timed word of the original system, and hence a trajectory $\chi(t)$ that is a solution to Problem \ref{problem}.

\subsection{Controller Synthesis} \label{sec:control_synthesis}

Fig. \ref{fig:solution_scheme} depicts a framework under which a sequence of feedback control laws $u(\chi, v)$ that guarantee the satisfaction of the MITL formula $\varphi$ can be computed. First, a TBA $\mathcal{A}$ that accepts all the timed runs satisfying the specification formula $\varphi$ is constructed. Second, a product between the WTS $\mathcal{T}$ given in Definition \ref{def:WTS_abstraction} and the TBA $\mathcal{A}$ is computed which gives the product WTS $\widetilde{\mathcal{T}}$. By performing graph search to the product WTS $\widetilde{\mathcal{T}}$, a timed run that satisfies the MITL formula $\varphi$ can be found. For more details regarding the control synthesis procedure we refer to our previous work \cite{alex_2016_acc, alex_automatica_2018}

\begin{proposition}
	The solution that it is obtained from the aforementioned controller synthesis procedure provides a sequence of
	feedback control laws $u(\chi, v)$ as in \eqref{eq:control_law_u} that guarantees the satisfaction of the formula $\varphi$ of the robot governed by dynamics as in \eqref{eq:kinematics}-\eqref{eq:unc_dynamics}, thus,  providing a solution to Problem \ref{problem}.
\end{proposition}

\begin{figure}[t!]
	\centering
	\includegraphics[scale = 0.7]{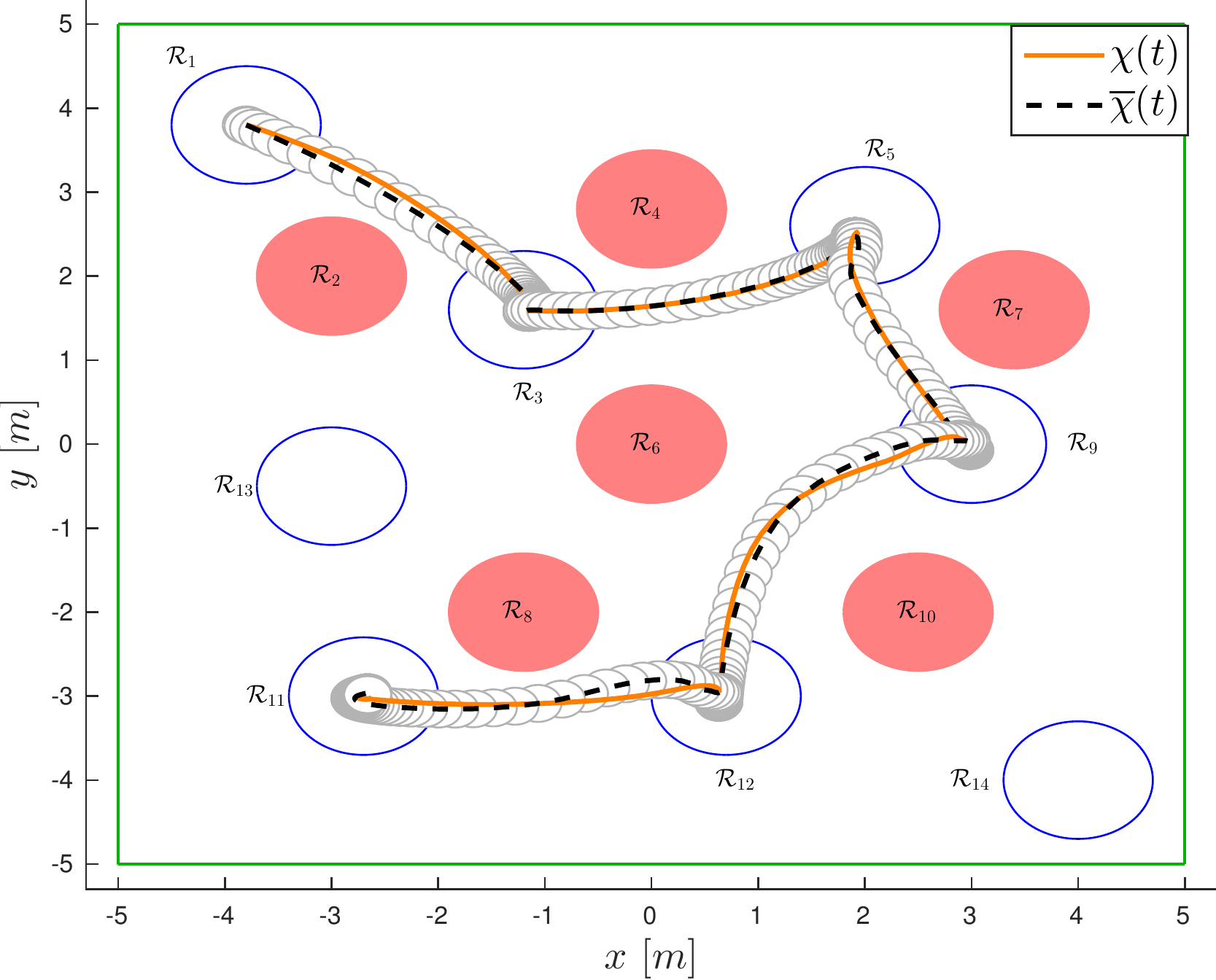}
	\caption{The evolution of the trajectory robot in the workspace $\mathcal{W}$. RoI and unsafe regions are depicted with blue and red color, respectively. The tube of the robot is depicted with light gray color. The real and the nominal trajectories $\chi(t)$ and $\overline{\chi}(t)$, respectively, are depicted with orange and dashed black color. The robot successfully satisfies the task $\varphi$ given in \eqref{eq:phi_simulation}.}\label{fig:workspace}
\end{figure}
\vspace{-5mm}
\section{Simulation Results} \label{sec:simulation_results}
For a simulation example, consider a robot operating in a workspace $\mathcal{W} = \{x,y \in \mathbb{R} : -5 \le x,y \le 5\} \subseteq \mathbb{R}^{2}$ with dynamics:
\begin{align*}
\dot{x} & = v_1, \\
\dot{y} & = v_2, \\
\dot{v}_1 & =  0.25 x^2+u_1+0.25 \cos(t), \\
\dot{v}_2 & =  \dfrac{0.1-0.1 e^{-x}}{1+ e^{-x}}+0.25 y^2+u_2+0.1 u_2^{3}+0.25 \sin(t),
\end{align*}
where $\chi = [x,y]^\top \in \mathbb{R}^{2}$, $v = [v_1, v_2]^\top \in \mathbb{R}^{2}$, $u = [u_1, u_2]^\top \in \mathbb{R}^2$, $d = [0.25 \cos(t), 0.25 \sin(t)]^\top$ and $\widetilde{d} = 0.25$. From \eqref{eq:J_lower_bound}, we get $J(\chi, v, u) = \begin{bmatrix} 1 & 0 \\ 0 & 1+0.3 u_2^2\end{bmatrix}$, which results in $\lambda_{\scriptscriptstyle \min}\left[\tfrac{J(\cdot)+J(\cdot)^\top}{2}\right] \ge \underline{J} = 1.$ The velocity and input constraints are:
\begin{align*}
\mathcal{V} & = \{v \in\ \mathbb{R}^2: -5 \le v_1, v_2 \le 5 \}, \\
\mathcal{U} & = \{u \in \mathbb{R}^2: -2.125 \le u_1, u_2 \le 2.125\},
\end{align*}
respectively. The Lipschitz constant is $L = 2.5$. The initial states of the robot are $\chi(0) = [-3.8, 3.8]^\top$ and $v(0) = [0,0]^\top$. The sampling time and the prediction horizon are set to $h = 0.1 \sec$, $T = 1.2 \sec$, respectively. %The control gains are chosen as $\rho = 1$, $k = 3.5$, which result to a tube with diameter $0.2$ and $0.4$ for states $e$ and $v$, respectively. 
The NMPC gains are set to $Q = P = I_{4}$, $R = 0.5 I_2$.

\begin{figure}[t!]
	\centering
	\includegraphics[scale = 0.70]{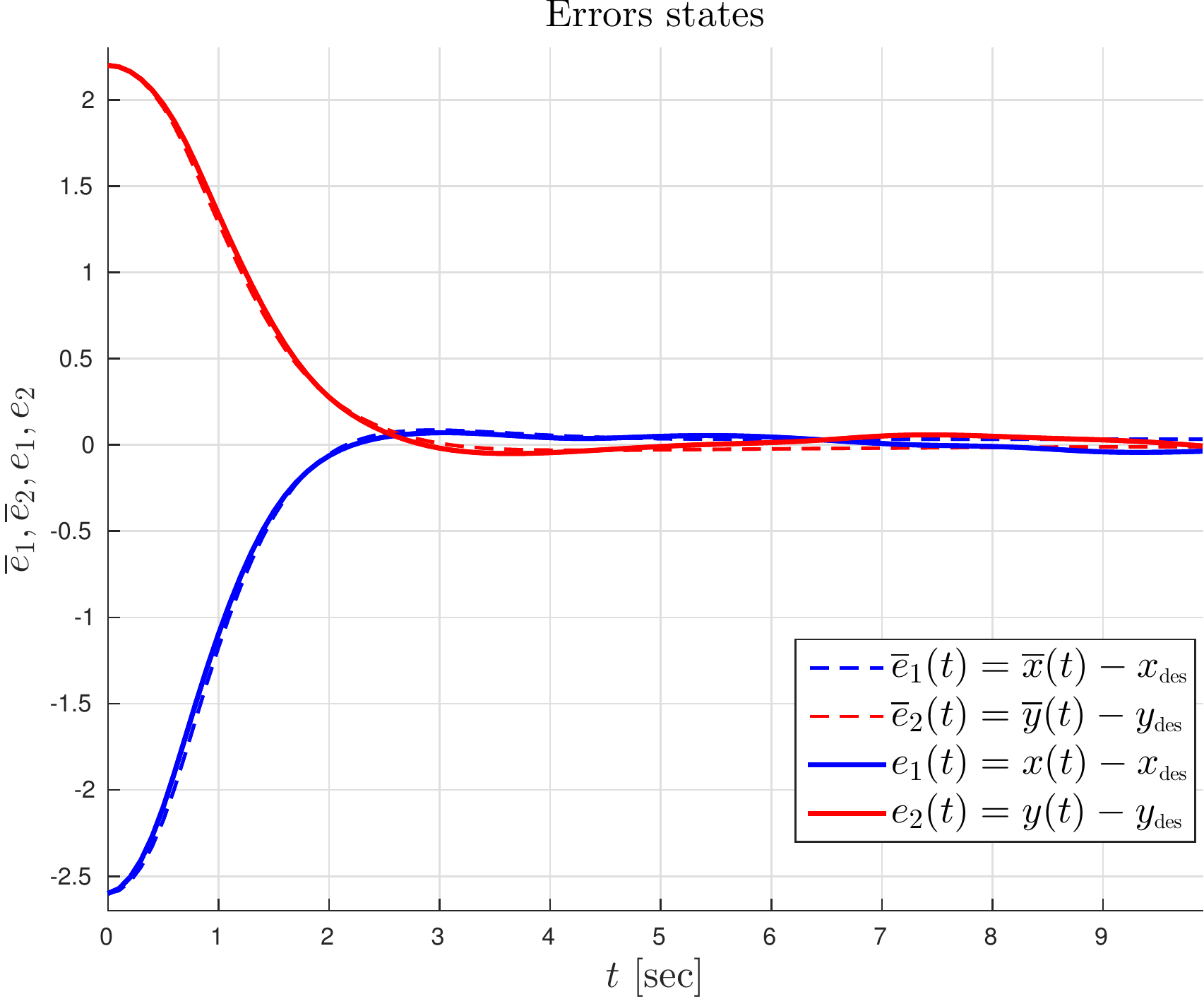}
	\caption{The evolution of the real errors $e_1(t) = x(t) - x_{\scriptscriptstyle \rm des}$, $e_2(t) = y(t) - y_{\scriptscriptstyle \rm des}$ and the corresponding nominal errors $\overline{e}_1(t) = \overline{x}(t) - x_{\scriptscriptstyle \rm des}$, $\overline{e}_2(t) = y(t) - x_{\scriptscriptstyle \rm des}$ for the transition between $\mathcal{R}_{1}$-$\mathcal{R}_{3}$.}\label{fig:errors}
\end{figure}

In the workspace we have $\mathfrak{m} = 14$ RoI with radius $p_m = 0.7$, $\forall m \in \mathcal{M}$ from which $6$ of them stand for unsafe regions that the robot is required not to visit (depicted with red color in Fig. \ref{fig:workspace}). The desired MITL formula is set as:
\begin{align} \label{eq:phi_simulation}
\hspace{-2mm} \varphi = \square_{[0,\infty)} \{\neg {\rm obs}\} \wedge \Diamond_{[6,12]} \{\rm goal_1\} \wedge \Diamond_{[20,30]} \{\rm goal_2\}, \hspace{-2mm}
\end{align}
over the set of atomic propositions $\Sigma = \{\rm obs, \rm goal_1, \rm goal_2\}$ and labeling function $L(\mathcal{R}_{5}) = \{\rm goal_1\}$, $L(\mathcal{R}_{11}) = \{\rm goal_2\}$, $L(\mathcal{R}_i) = \{\rm obs\}$, $i \in \{2,4,6,7,8,10\}$ and $L(\mathcal{R}_i) = \emptyset$, $i \in \{1,$ $3$, $9$, $12$, $13$, $14\}$. Fig. \ref{fig:workspace} depicts the workspace with RoI, unsafe regions, the nominal trajectory of the robot (orange color), the real trajectory of the robot (black color) and the tube centered along the nominal trajectory. By using Algorithm $1$ the time duration of the transitions between RoI $\mathcal{R}_1$-$\mathcal{R}_3$, $\mathcal{R}_3$-$\mathcal{R}_5$, $\mathcal{R}_5$-$\mathcal{R}_9$, $\mathcal{R}_9$-$\mathcal{R}_{12}$ and $\mathcal{R}_{12}$-$\mathcal{R}_{11}$ are $\mathfrak{t}(\mathcal{R}_1, \mathcal{R}_3) = 5.2 \sec$, $\mathfrak{t}(\mathcal{R}_3, \mathcal{R}_5) = 6.1 \sec$, $\mathfrak{t}(\mathcal{R}_5, \mathcal{R}_9) = 4.5 \sec$, and $\mathfrak{t}(\mathcal{R}_9, \mathcal{R}_{12}) = 7.1 \sec$ and $\mathfrak{t}(\mathcal{R}_{12}, \mathcal{R}_{11}) = 4.8 \sec$, respectively. According to Fig. \ref{fig:workspace}, the robot never visits the unsafe RoI. Furthermore, it navigates to goal RoI $\mathcal{R}_{5}$ and $\mathcal{R}_{11}$ at time $11.3 \sec$, $27.7 \sec$, respectively, which results in a successful satisfaction of $\varphi$. Thus, $\chi(t) \models \varphi$, $\forall t \ge 0$. The error signals for the transition between RoI $\mathcal{R}_1$-$\mathcal{R}_3$ are depicted in Fig. \ref{fig:errors}. The evolution of the velocities $v_1(t)$ and $v_2(t)$ is presented in Fig. \ref{fig:velocities}. The control effort for the transition from $\mathcal{R}_1$ to $\mathcal{R}_3$ is presented in Fig. \ref{fig:inputs}.  Finally, Fig. \ref{fig:tube} shows a more detailed zoom in the tube of the trajectory of the robot.

\begin{figure}[t!]
	\centering
	\includegraphics[scale = 0.70]{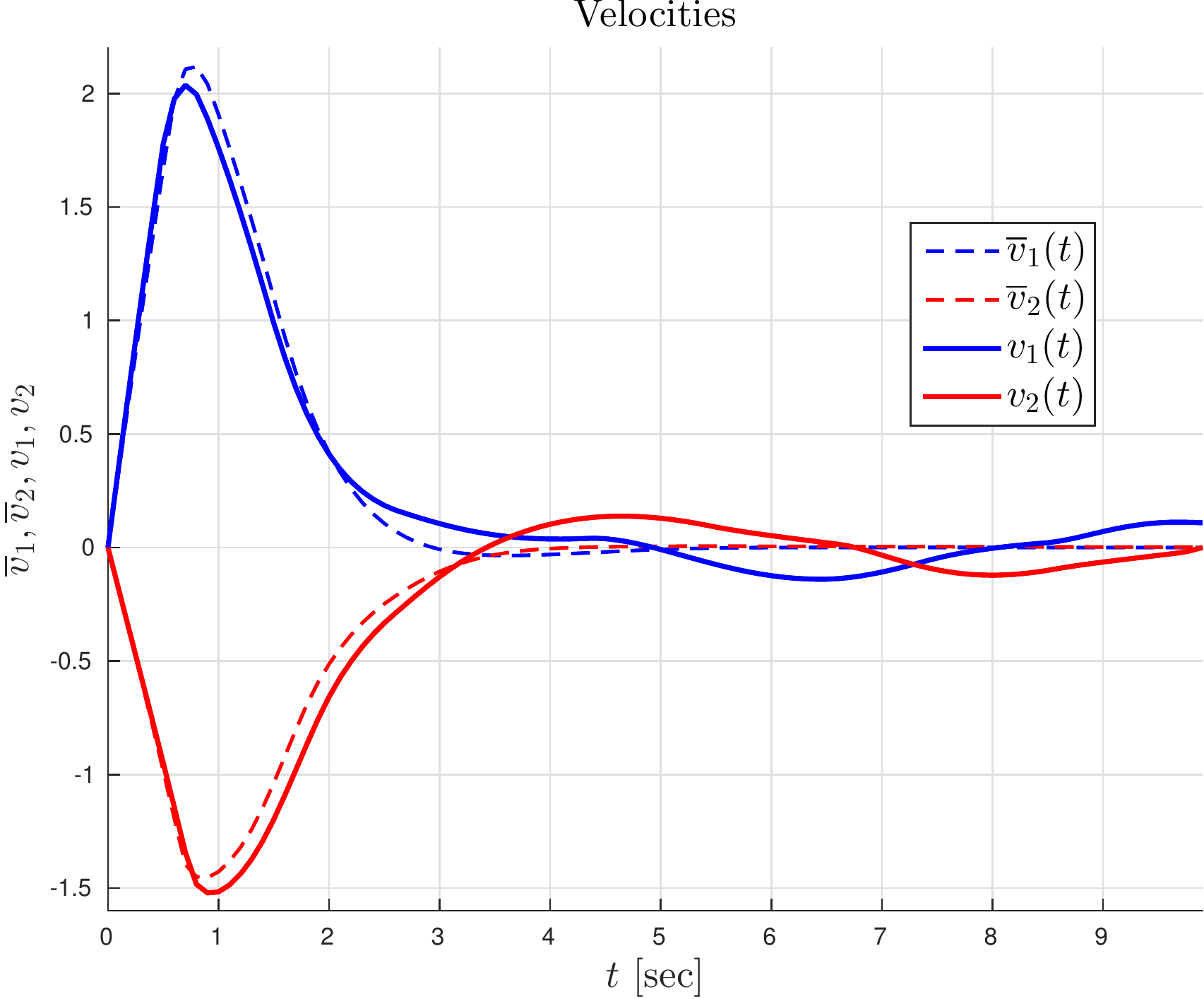}
	\caption{The real and the nominal velocity signals $v_1(t)$, $v_2(t)$ and $\overline{v}_1(t)$, $\overline{v}_2(t)$, respectively for the transition between $\mathcal{R}_{1}$-$\mathcal{R}_{3}$. It holds that $v_i(t) \in \mathcal{V}$, $i \in \{1,2\}$.}\label{fig:velocities}
\end{figure}

The simulation was conduced in MATLAB R2015a by using optimization tools found in \cite{grune2016nonlinear}. It takes $34.2 \sec$ on a laptop with $4$ cores, i7-$2.80$ GHz CPU and~$16$GB~of~RAM.

\begin{figure}[t!]
	\centering
	\includegraphics[scale = 0.70]{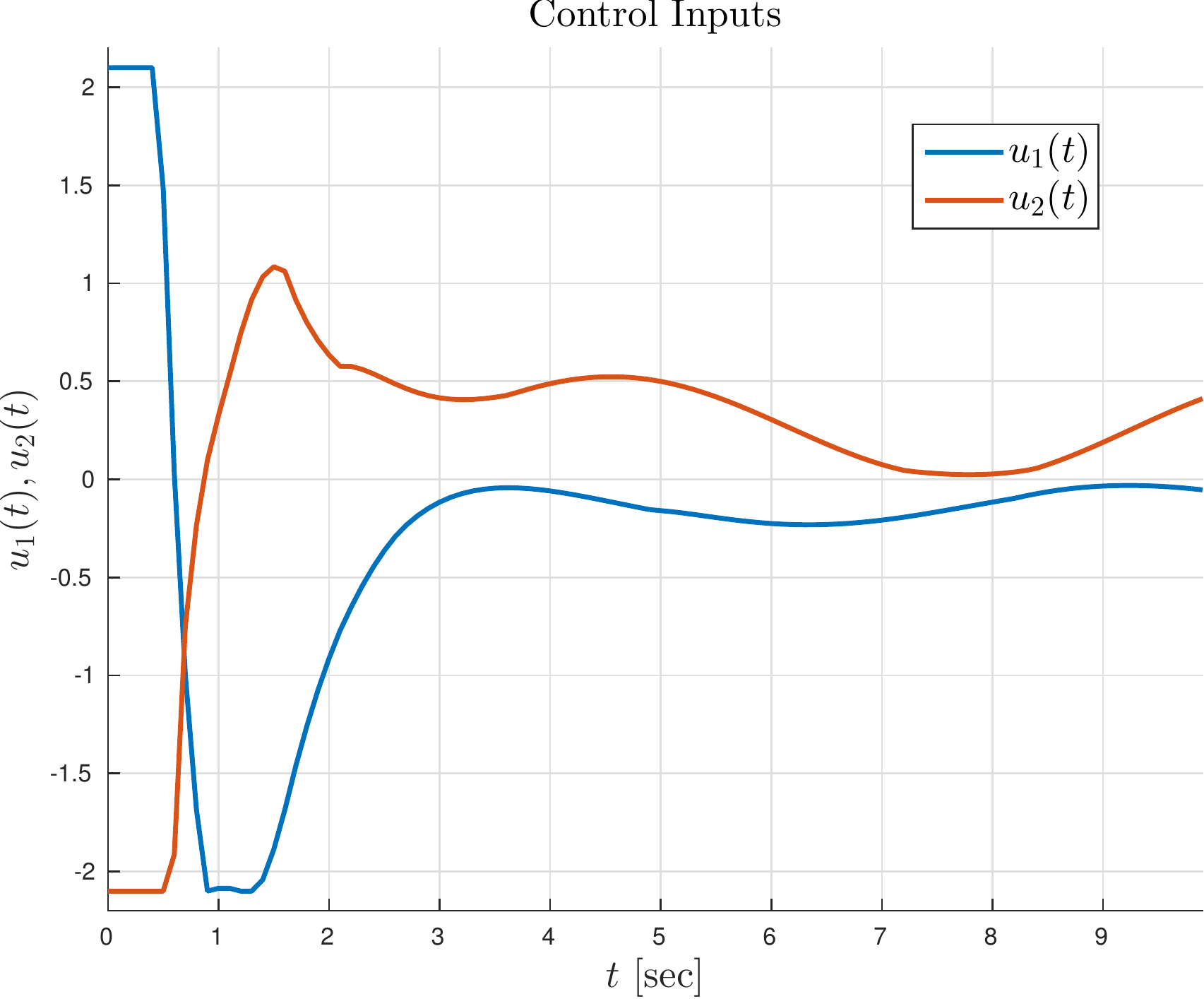}
	\caption{The control input signals $u_1(t)$, $u_2(t)$ for navigating the robot between $\mathcal{R}_1$-$\mathcal{R}_3$. It holds that $u_i(t) \in \mathcal{U}$, $i \in \{1,2\}$.}\label{fig:inputs}
\end{figure}

\section{Conclusions and Future Research} \label{sec:conclusions}

This paper addresses the problem of robot navigation under time constraints given in MITL. The robot is operating in a workspace which is a subset of $\mathbb{R}^{n}$ and it is modeled by nonlinear non-affine kinematics/dynamics. A robust tube-based NMPC scheme which guarantees robust transitions between RoI of the workspace is proposed. By utilizing algorithmic and verification tools from previous work, a framework of controller synthesis which computes the sequence of feedback control laws that provably satisfies the given formula is provided. Future research will be devoted towards extending the current framework to multi-robot systems with event-triggered control communication laws. 

\begin{figure}[t!]
	\centering
	\includegraphics[scale = 0.70]{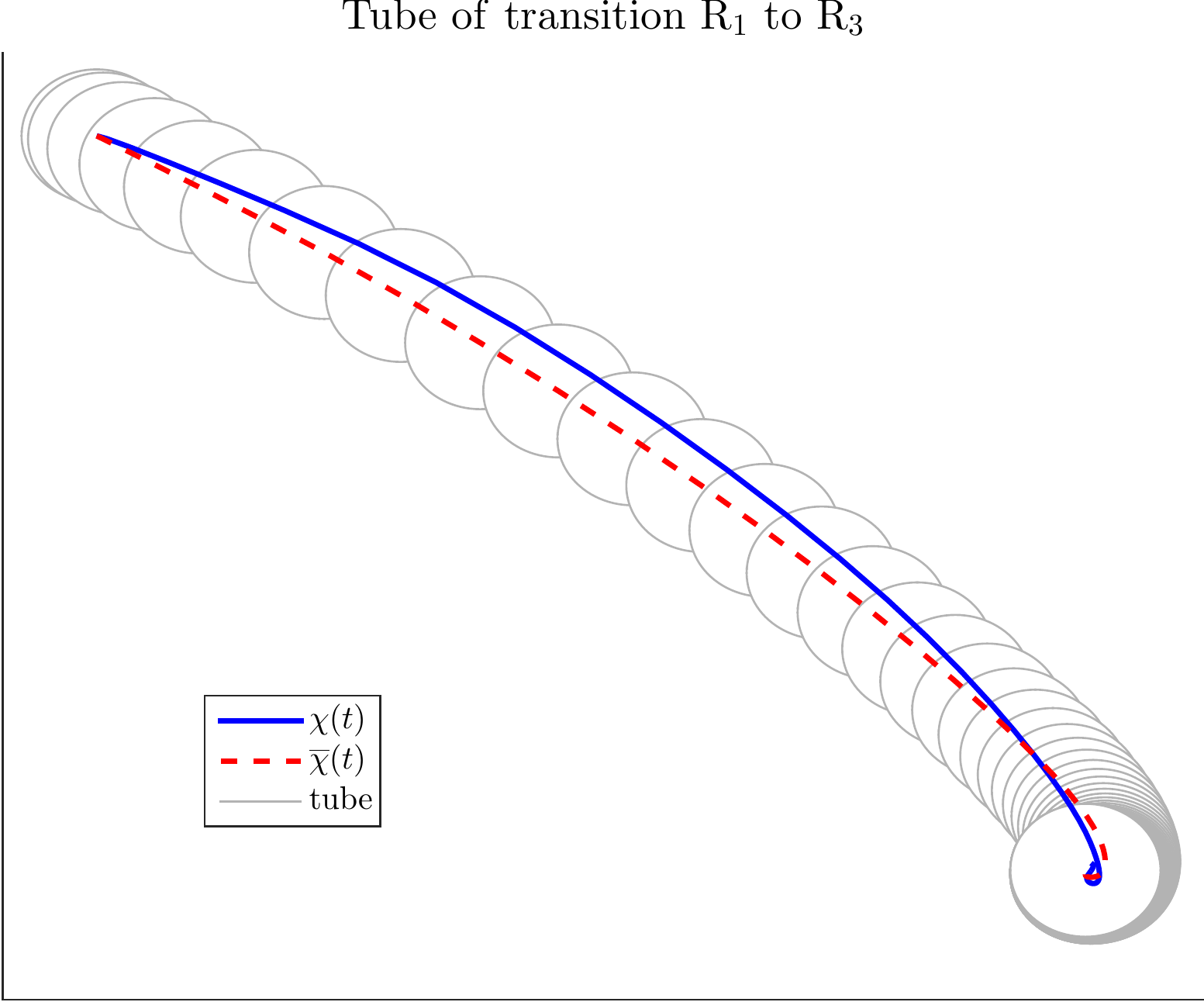}
	\caption{The tube (gray color) from the transition from $\mathcal{R}_1$ to $\mathcal{R}_3$ centered along the nominal trajectory $\overline{\chi}(t)$ (red dashed color) and the real trajectory $\chi(t)$ (blue color). The real trajectory $\chi(t)$ remains always within the tube.}\label{fig:tube}
\end{figure}

\bibliographystyle{ieeetr}   % Include this if you use bibtex 
\bibliography{references}

\clearpage 

\begin{wrapfigure}{l}{25mm} 
	\includegraphics[width=1in,height=1.25in,clip,keepaspectratio]{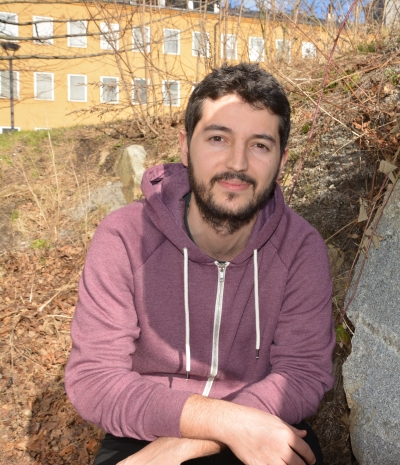}
\end{wrapfigure}\par
\textbf{Alexandros Nikou} was born in 1988. He received the Diploma in Electrical and Computer Engineering in 2012 and the M.Sc. in Automatic Control in 2014, both from National Technical University of Athens (NTUA), Greece. He is currently a PhD student at the Department of Automatic Control, School of Electrical Engineering and Computer Science, KTH Royal Institute of Technology, Stockholm, Sweden. His current research interests include Multi-Agent Systems Control, Distributed Nonlinear Model Predictive Control and Formal Methods in Control.\par

\vspace{5mm}

\begin{wrapfigure}{l}{25mm} 
	\includegraphics[width=1.6in,height=1.25in,clip,keepaspectratio]{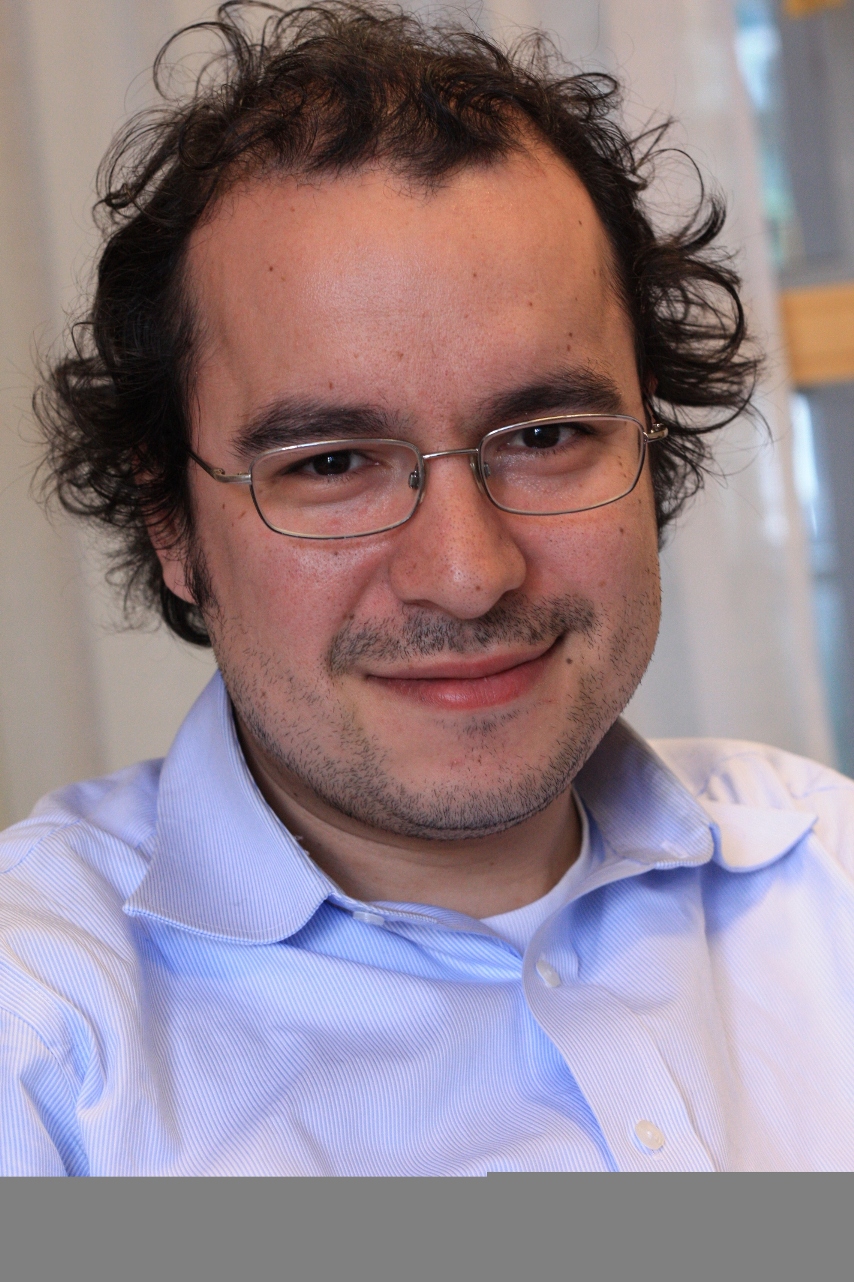}
\end{wrapfigure}\par
\textbf{Dimos Dimarogonas} was born in Athens, Greece, in 1978. He received the Diploma in Electrical and Computer Engineering in 2001 and the PhD in Mechanical Engineering in 2007, both from National Technical University of Athens (NTUA), Greece. Between May 2007 and February 2009, he was a Postdoctoral Researcher at the Department of Automatic Control, School of Electrical Engineering and Computer Science, Royal Institute of Technology (KTH), Stockholm, Sweden. Between February 2009 and March 2010, he was a Postdoctoral Associate at the Laboratory for Information and Decision Systems (LIDS) at the Massachusetts Institute of Technology (MIT), Boston, MA, USA. He is currently Professor at the Department of Automatic Control, KTH Royal Institute of Technology, Stockholm, Sweden. His current research interests include Multi-Agent Systems, Hybrid Systems and Control, Robot Navigation and Networked Control. He serves in the Editorial Board of Automatica, the IEEE Transactions on Automation Science and Engineering and the IET Control Theory and Applications and is a Senior member of IEEE and the Technical Chamber of Greece.\par
\end{document}